\documentclass[sigconf]{acmart}

\usepackage{booktabs} 
\usepackage{IEEEtrantools}
\usepackage{latexsym,amsfonts,amsmath}
\usepackage{mathrsfs}
\usepackage{graphicx}
\usepackage{wrapfig}
\usepackage{paralist}
\usepackage{dsfont}
\usepackage{eso-pic}
\usepackage{epsfig}
\usepackage{ragged2e}
\usepackage{mathtools}
\usepackage{optidef}
\usepackage{epstopdf}
\usepackage{lipsum}
\usepackage{varwidth}
\usepackage{fancyhdr}
\usepackage{subfig}
\usepackage{multirow}
\newcommand*{\myalign}[2]{\multicolumn{1}{#1}{#2}}

\newtheorem{problem}{Problem}
\newtheorem{algorithm}{Algorithm}

\newtheorem{remark}{Remark}

\usepackage{tikz}
\usetikzlibrary{arrows,calc,decorations.markings,math,arrows.meta}
\usepackage{pgfplots}

\usepackage{algorithm}
\usepackage{algorithmic}

\usepackage{xcolor}         
\definecolor{rangyek}{RGB}{0, 75, 255}
\definecolor{rangdo}{RGB}{239, 62, 91}

\DeclareRobustCommand\sampleline[1]{%
	\tikz\draw[#1] (0,0) (0,\the\dimexpr\fontdimen22\textfont2\relax)
	-- (2em,\the\dimexpr\fontdimen22\textfont2\relax);%
}
\definecolor{START}{rgb}{0.4660 0.6740 0.1880}
\definecolor{TARGET}{rgb}{0 0.4470 0.7410}
\definecolor{OBSTACLES}{rgb}{0.93,0.69,0.13}

\definecolor{START1}{rgb}{0, 1, 0}
\definecolor{TARGET1}{rgb}{0.3010 0.7450 0.9330}
\definecolor{OBSTACLES1}{rgb}{1, 0, 0}

\usepackage[utf8]{inputenc} 
\usepackage[T1]{fontenc}    
\usepackage{hyperref}       
\hypersetup{
	colorlinks = true,
	linkcolor = rangyek,
	citecolor = rangdo
}
\usepackage{url}            
\usepackage{booktabs}       
\usepackage{amsfonts}       
\usepackage{nicefrac}       
\usepackage{microtype}      

\usepackage{stackengine}
\usepackage{scalerel}

\usepackage[many]{tcolorbox}
\usetikzlibrary{calc}
\tcbuselibrary{skins}

\makeatletter
\newsavebox\myboxA
\newsavebox\myboxB
\newlength\mylenA

\newcommand*\xoverline[2][0.75]{%
	\sbox{\myboxA}{$\m@th#2$}%
	\setbox\myboxB\null
	\ht\myboxB=\ht\myboxA%
	\dp\myboxB=\dp\myboxA%
	\wd\myboxB=#1\wd\myboxA
	\sbox\myboxB{$\m@th\overline{\copy\myboxB}$}
	\setlength\mylenA{\the\wd\myboxA}
	\addtolength\mylenA{-\the\wd\myboxB}%
	\ifdim\wd\myboxB<\wd\myboxA%
	\rlap{\hskip 0.5\mylenA\usebox\myboxB}{\usebox\myboxA}%
	\else
	\hskip -0.5\mylenA\rlap{\usebox\myboxA}{\hskip 0.5\mylenA\usebox\myboxB}%
	\fi}
\makeatother

\newtcolorbox{resp}[1][]{%
	enhanced jigsaw,%
	colback=gray!5!white,%
	colframe=gray!80!black,%
	size=small,%
	boxrule=1pt,%
	halign title=flush center,%
	coltitle=black,%
	breakable,%
	drop shadow=black!50!white,%
	attach boxed title to top left={xshift=1cm,yshift=-\tcboxedtitleheight/2,yshifttext=-\tcboxedtitleheight/2},%
	minipage boxed title=3cm,%
	boxed title style={%
		colback=white,%
		size=fbox,%
		boxrule=1pt,%
		boxsep=2pt,%
		underlay={%
			\coordinate (dotA) at ($(interior.west) + (-0.5pt,0)$);
			\coordinate (dotB) at ($(interior.east) + (0.5pt,0)$);
			\begin{scope}[gray!80!black]
				\fill (dotA) circle (2pt);
				\fill (dotB) circle (2pt);
			\end{scope}
		}%
	},%
	#1%
}

\newtcolorbox{mybox}[2][]{enhanced,
	attach boxed title to top left={xshift=1cm,yshift=-2mm},
	fonttitle=\bfseries,varwidth boxed title=0.7\linewidth,
	colbacktitle=gray!45!white,coltitle=gray!10!black,colframe=gray!50!black,
	interior style={top color=gray!5!white,bottom color=gray!0!white},
	boxed title style={boxrule=0.75mm,colframe=white,
		borderline={0.1mm}{0mm}{gray!50!black},
		borderline={0.1mm}{0.75mm}{gray!50!black},
		interior style={top color=gray!30!white,bottom color=gray!5!white,
			middle color=gray!50!white},
		drop fuzzy shadow},
	title={#2},#1}

\usepackage{xspace}
\newcommand{\R}{{\mathbb{R}}}

\newcommand{\ie}{{\it i.e.}}

\newcommand{\Let}{:=}

\newcommand{\MATLAB}{\textsc{Matlab}\xspace}

\usepackage{xcolor}
\definecolor{lightblue}{rgb}{0.30, 0.75, 0.93}
\definecolor{lightblue}{rgb}{0.30, 0.75, 0.93}
\definecolor{mycolor}{rgb}{0, 0.45, 0.75}
\definecolor{mycolor1}{rgb}{0.39, 0.83, 0.07}
\definecolor{fluorescentpink}{rgb}{1.0, 0.08, 0.58}
\definecolor{royalblue(web)}{rgb}{0.25, 0.41, 0.88}
\definecolor{vividcerise}{rgb}{0.85, 0.11, 0.51}
\definecolor{tangelo}{rgb}{0.98, 0.3, 0.0}
\definecolor{persiangreen}{rgb}{0.0, 0.65, 0.58}
\definecolor{lemon}{rgb}{1.0, 0.97, 0.0}
\definecolor{gdash}{rgb}{0.10,0.82,0.10}

\definecolor{fluorescentpink}{rgb}{1.0, 0.08, 0.58}
\definecolor{D5Mr}{rgb}{0.62, 0.0, 1.0}
\definecolor{cUdg}{rgb}{0.3, 0.73, 0.09}
\definecolor{QKVd}{rgb}{0.03, 0.57, 0.82}
\definecolor{RE12}{rgb}{0.48, 0.25, 0.0}

\begin{document}
	
	\title[From Data to Global Asymptotic Stability of Unknown Large-Scale Networks]{From Data to Global Asymptotic Stability of Unknown Large-Scale Networks with Provable Guarantees}
	
	\author{Mahdieh Zaker}
	\affiliation{%
			\institution{School of Computing}
		   \institution{Newcastle University}
		\country{}}
	\email{mahdieh.zaker@newcastle.ac.uk}
	
    \author{Amy Nejati}
	\affiliation{%
			\institution{School of Computing}
		   \institution{Newcastle University}
		\country{}}
	\email{amy.nejati@newcastle.ac.uk}
    
	\author{Abolfazl Lavaei}
	\affiliation{%
			\institution{School of Computing}
		   \institution{Newcastle University}
		\country{}}
	\email{abolfazl.lavaei@newcastle.ac.uk}
	
	\begin{abstract}
		
		We offer a \emph{compositional data-driven} scheme for synthesizing controllers that ensure global asymptotic stability (GAS) across large-scale interconnected networks, characterized by \emph{unknown} mathematical models. In light of each network's configuration composed of numerous subsystems with smaller dimensions, our proposed framework gathers data from each subsystem's trajectory, enabling the design of \emph{local controllers} that ensure \emph{input-to-state stability (ISS)} properties over subsystems, signified by ISS Lyapunov functions. To accomplish this, we require only a \emph{single input-state} trajectory from each unknown subsystem up to a specified time horizon, fulfilling certain rank conditions. Subsequently, under \emph{small-gain compositional} reasoning, we leverage ISS Lyapunov functions derived from data to offer a \emph{control Lyapunov function (CLF)} for the interconnected network, ensuring GAS certificate over the network. We demonstrate that while the computational complexity for designing a CLF increases \emph{polynomially} with the network dimension using sum-of-squares (SOS) optimization, our compositional data-driven approach significantly mitigates it to \emph{linear} with respect to the number of subsystems. We showcase the efficacy of our data-driven approach over a set of benchmarks, involving physical networks with diverse interconnection topologies.
		
	\end{abstract}
	
	
	\maketitle
	
	\section{Introduction}\label{sec:intro}
	In recent years, significant focus has been directed towards large-scale interconnected networks, driven by their capacity to model a wide array of real-world applications, including, but not limited to, biological networks, transportation systems, and interconnected power grids. The efficiency of an interconnected network relies on how its subsystems interact and contribute to the overall behavior of the network.  When an individual subsystem experiences failure or disturbance, it can propagate disruptions across other subsystems, causing significant upheavals in the network's operation.
	
	Input-to-state stability has been essentially employed to investigate robust behaviors of interconnected networks, influenced by disturbances acting as adversarial inputs~\citep{nesic2002integral,agrachev2008input,isidori2013nonlinear}. In particular, the influence of any subsystem on others within an interconnected network can be effectively captured by ISS Lyapunov functions, while constructing interaction gain functions. These functions are instrumental in delineating the interconnection structure of the network and capturing the effect of neighboring subsystems. According to the \emph{small-gain reasoning}, if the constructed gains satisfy certain compositional conditions, the overall network would exhibit stability~\citep{jiang1994small, jiang1996lyapunov, dashkovskiy2007iss, dashkovskiy2010small}.
	
	\noindent\textbf{Primary Challenge.} To conduct stability analysis over interconnected networks using ISS properties, \emph{precise mathematical models} of subsystems are essential. However, this situation is not typical in numerous real-world applications, where system's dynamics are often unknown. To address this challenge, data-driven approaches have emerged in two distinct categories: \emph{indirect and direct}~\citep{dorfler2022bridging}. Indirect data-driven approaches aim to approximate unknown dynamics through identification techniques; nonetheless, obtaining an accurate mathematical model can be computationally demanding, particularly for complex nonlinear systems~\citep{hou2013model}. Even if a model can be identified through system identification approaches, designing a controller that guarantees the input-to-state stability of the identified model remains challenging. Hence, the complexity arises from two layers: (i) model identification, and (ii) controller synthesis using conventional model-based approaches. In contrast, \emph{direct} data-driven techniques aim to circumvent system identification and directly offer formal analysis on the stability of the system solely based on observed data~\citep{dorfler2022bridging}.
	
	\noindent{\bf Background Research.} There have been a few studies based on data-driven frameworks to provide stability certificates for unknown dynamical systems. In this context, \citep{depersis2020tac} and \citep{guo2021data} propose a data-driven scheme for the stability analysis of persistently excited linear time-invariant and nonlinear polynomial systems, respectively. In \citep{berberich2020data}, a robust data-driven model-predictive control is introduced for linear time-invariant systems. The study in \citep{taylor2021towards} develops a data-driven approach for robust control synthesis of nonlinear systems considering model uncertainty. Additionally, \citep{harrison2018control} introduces a data-driven approach for learning controllers capable of rapid adaptation to unknown dynamics. In \citep{vzikelic2024compositional}, a methodology is introduced for learning a composition of neural network policies in stochastic environments, providing a formal probabilistic certification ensuring the fulfillment of a specification over the policy's behavior. Data-driven stability analysis for unknown systems is presented for switched systems in \citep{kenanian2019data}, continuous-time systems in \citep{boffi2021learning}, and discrete-time systems in \cite{lavaei2022data}. The work~\cite{zhou2022neural} presents a framework for stabilizing unknown nonlinear systems by jointly learning a neural Lyapunov function and a nonlinear controller, ensuring theoretical stability guarantees through the use of Satisfiability Modulo Theories (SMT) solvers. By applying overapproximation techniques to characterize the set of polynomial dynamics that align with the noisy measured data, \cite{chen2024data} constructs an ISS Lyapunov function and a corresponding ISS control law for unknown nonlinear input-affine systems with polynomial dynamics. A data-driven method for synthesizing the \emph{incremental} ISS controllers and certifying the incremental ISS property for unknown nonlinear polynomial systems has recently been explored in~\cite{zaker2024certified}.
	
	While studies \citep{depersis2020tac,guo2021data,berberich2020data,taylor2021towards,harrison2018control,vzikelic2024compositional,kenanian2019data,boffi2021learning,lavaei2022data,zhou2022neural,chen2024data,zaker2024certified} focus on stability analysis and controller synthesis using data, their methods are specifically tailored to \emph{monolithic systems} with low dimensions and cannot be directly applied to large-scale networks due to the \emph{sample complexity} problem. Some limited recent efforts aim to address stability analysis for higher-dimensional systems. In particular, the work \cite{sperl2023approximation} explores how deep neural networks can approximate control Lyapunov functions for high-dimensional systems, aiming to circumvent the curse of dimensionality. Utilizing the compositional structure inherent in interconnected nonlinear systems, \cite{liu2024compositionally} proposes a compositional method to train and validate neural Lyapunov functions for stability analysis. In \cite{zhang2023compositional}, a compositional approach is employed to provide neural certificates based on ISS Lyapunov functions and controllers for networked dynamical systems. The work \cite{lavaei2023ISS} provides a data-driven stability certificate for interconnected networks.
	
	While \cite{sperl2023approximation,liu2024compositionally,zhang2023compositional} provide compositional techniques to ensure stability in large-scale systems, they each assume \emph{known system dynamics}, in contrast to our work which focuses on unknown dynamics—a more challenging scenario that is often the case in reality. Additionally, the ISS Lyapunov functions constructed by neural networks in \cite{zhang2023compositional} are merely candidate functions without a correctness guarantee, necessitating an \emph{additional verification step} using SMT solvers such as Z3~\cite{de2008z3}, dReal~\cite{gao_-complete_2012} or MathSat~\cite{cimatti_mathsat5_2013}. However, SMT solvers may fail to terminate \cite{wongpiromsarn2015automata} or encounter scalability issues depending on the neural network size or system complexity (see~\cite[Section 5.2]{abate2022neural}). In contrast, our method provides ISS Lyapunov functions with a correctness guarantee in a single step. Furthermore, the compositional condition in \cite{zhang2023compositional} relies on a stringent symmetry assumption regarding the network structure—a restrictive requirement that our approach does not impose. 
	
   In contrast to \citep{lavaei2023ISS}, their approach employs a \emph{scenario-based method }that requires data to be \emph{independent and identically distributed (i.i.d.)}. This limitation allows only one input-output data pair per trajectory~\citep{calafiore2006scenario}, necessitating the collection of \emph{multiple independent trajectories} (up to millions in practical scenarios) to achieve the desired confidence level, based on a known closed-form relationship among them. However, our approach requires only a \emph{single input-state trajectory} from each subsystem to derive our results. Additionally,  \citep{lavaei2023ISS} restricts the system class to be \emph{homogeneous}, which does not apply to many practical scenarios. Our results, however, apply to nonlinear systems with polynomial dynamics. Furthermore, although \citep{lavaei2023ISS} solely provides \emph{verification analysis} over an interconnected network, our work offers controller design—a more challenging and useful task—to enforce the global asymptotic stability property over the network.

   Additionally, while we provide the GAS certificate for interconnected networks of nonlinear subsystems by incorporating the influence of subsystems on each other within the network topology via adversarial inputs, \citep{depersis2020tac} considers only linear time-invariant systems using a monolithic approach, which cannot enforce the GAS property for high-dimensional systems (exceeding five states). In contrast, our results apply to systems with up to $2000$ states (cf. Table~\ref{tab:benchmark}). Notably, \citep{depersis2020tac} does not also investigate the ISS property. Furthermore, although it examines robustness to measurement noise, the analysis assumes bounded noise, whereas our results ensure robustness against \emph{unbounded} adversarial inputs.
	
	\noindent\textbf{Original Contributions.} Our work introduces an innovative framework in a \emph{compositional data-driven} manner for synthesizing controllers over an \emph{unknown} interconnected network, while ensuring its global asymptotic stability. 
	Our methodology collects a single trajectory from each unknown subsystem, facilitating the design of a local controller alongside an ISS Lyapunov function for each subsystem based on data. By leveraging the \emph{small-gain compositional} reasoning, we then compose \emph{data-driven} ISS Lyapunov functions and construct a \emph{control Lyapunov function } for the interconnected network, thereby ensuring its GAS certificate.  
	While the computational complexity of CLF design, using SOS optimization, escalates \emph{polynomially} with the network dimension, our compositional data-driven approach offers a \emph{linear complexity} concerning the number of subsystems.  We employ several benchmarks to demonstrate the applicability of our data-driven approach over some physical networks with different interconnection topologies.
	
	\section{Problem description}\label{Problem_Description}
	
	\subsection{Notation}
	We employ the following notation throughout the paper. Sets of real, positive and non-negative real numbers are denoted by $\R,\R^+$, and $\R^+_0$, respectively. We denote sets of non-negative and positive integers by $\mathbb{N} := \{0,1,2,\ldots\}$ and $\mathbb{N}^+:=\{1,2,\ldots\}$, respectively. Given $N$ vectors $x_i \in \R^{n_i}$, $x=[x_1;\ldots;x_N]$ denotes the corresponding column vector of dimension $\sum_i n_i$. Moreover, we use $[x_1\;\; \ldots\;\; x_N]$ and $[A_1\;\;\ldots\;\;A_N]$ to represent horizontal concatenation of vectors $x_i\in\R^n$ and matrices $A_i\in\R^{n\times m}$ to form $n\times N$ and $n\times mN$ matrices, respectively. We use $\Vert\cdot\Vert$ to signify the Euclidean norm of a vector $x\in\R^{n}$ and the induced $2$-norm of a matrix $A\in\R^{n\times m}$.  For any square matrix $P$, $\lambda_{\min}(P)$ and $\lambda_{\max}(P)$ denote minimum and maximum eigenvalues of $P$, respectively.
	We use $P\succ 0$ to represent that a \emph{symmetric} matrix $P\in\R^{n\times n}$ is positive definite, \emph{i.e.,} all its eigenvalues are positive. A (block) diagonal matrix in $\R^{N\times{N}}$ with  diagonal matrix entries $(A_1,\ldots,A_N)$ and scalar entries $(a_1,\ldots,a_N)$ is denoted by $\mathsf{blkdiag}(A_1,\ldots,A_N)$ and $\mathsf{diag}(a_1,\ldots,a_N)$, respectively. The construction of a matrix with elements $a_{ij}$ in the $i$-th row and $j$-th column is denoted by $\{a_{ij}\}$. A column vector in $\R^{n}$ with all elements equal to one is represented by $\mathds{1}_n$. The transpose of a matrix $P$ is denoted by $P^\top$, while its \emph{left (right)} pseudoinverse is represented as $P^\dagger$ ($P^\ddagger$). We denote an $n\times n$ identity matrix by $\mathds I_n$.
	A function $\beta: \mathbb{R}_{0}^+ \rightarrow \mathbb{R}_{0}^+$ is classified as a $\mathcal{K}$ function if it is continuous, strictly increasing, and satisfies $\beta(0)=0$, while it belongs to the class $\mathcal{K}_\infty$ if it approaches infinity as its argument tends to infinity. A function $\beta: \mathbb{R}_{0}^+ \times \mathbb{R}_{0}^+ \rightarrow \mathbb{R}_{0}^+$ is considered to be in the class $\mathcal{KL}$ if, for each fixed $s$, $\beta(r,s)$ is a $\mathcal{K}$ function with respect to $r$, and for each fixed $r > 0$,  $\beta(r,s)$ is decreasing with respect to $s$, and it approaches 0 as $s$ tends to infinity.
	
	\subsection{Individual Subsystems}\label{systems1}
	In this work, we characterize each subsystem as a continuous-time nonlinear polynomial system, as elucidated in the subsequent definition.
	\begin{definition}
		A continuous-time nonlinear polynomial system (ct-NPS) is described by  
		\begin{align}\label{eq:subsystem}
			\Sigma_i\!: \dot x_i=A_i\mathcal F_i(x_i) + B_iu_i + D_iw_i,
		\end{align}
		where $\mathcal F_i(x_i) \in \R^{N_i}$ is a vector of monomials in states $x_i\in \R^{n_i}$, with $\mathcal F_i(0)=0$, $A_i \in \R^{n_i\times N_i}$ is the system matrix, $B_i \in \R^{n_i\times m_i}$ is the control matrix, and $D_i \in \R ^{n_i \times \sigma_i}$ is the adversarial matrix with $\sigma_i = \sum_{j = 1,j\neq i}^{M} n_j$, where $M$ is the number of subsystems. In addition, $u_i\in \R^{m_i}$ denotes the \emph{control} input, and $w_i \in \R^{\sigma_i}$ represents the \emph{adversarial} input of ct-NPS. We employ the tuple $\Sigma_i\!=\!(A_i,B_i, D_i, \R^{n_i},\R^{N_i},\R^{m_i}, \R^{\sigma_i})$ to denote the subsystem in~\eqref{eq:subsystem}.
	\end{definition}
	
	In the context of our data-driven work, both matrices $A_i$ and $B_i$ are \emph{unknown}, while matrix $D_i$ is considered known as it captures the weights of interconnections among subsystems, which are often a-priori known in interconnected networks. Moreover, we let $\mathcal F_i(x_i)$ include all monomials up to a maximum degree, which can be inferred in ascending order based on physical insights into the unknown system. This approach ensures that $\mathcal F_i(x_i)$ potentially includes all monomial terms present in the actual unknown subsystem.
	
	Given that the primary focus of this work is to conduct stability analysis over interconnected networks consisting of individual subsystems in the form of~\eqref{eq:subsystem}, in the following subsection, we present the network structure and elucidate how subsystems can be interconnected.
	
	\subsection{Interconnected Network}\label{In-Net}
	Here, we provide a formal definition of interconnection between subsystems $\Sigma_i$. Let us consider an interconnected network $\Sigma$, comprising $M \in \mathbb{N}^+$ subsystems $\Sigma_i$, with adversarial inputs and their corresponding matrices partitioned as
		\begin{align}\label{eq:partitioned w}
			&w_i  =[w_{i 1} ; \ldots ; w_{i(i-1)} ; w_{i(i+1)} ; \ldots ; w_{i M}],\\
			&D_i = [D_{i 1} \;\; \ldots \;\; D_{i(i-1)} \;\; D_{i(i+1)} \;\; \ldots \;\; D_{i M}],
		\end{align}
		where $D_{ij}\in\R^{n_i\times n_j}$. We assume the dimension of $w_{ij}$ matches that of $x_j$, a well-defined assumption in the context of small-gain reasoning. If there is no connection from subsystem $\Sigma_j$ to $\Sigma_i$, it implies that the corresponding adversarial input is identically zero, \ie, $w_{ij}\equiv 0$, otherwise, $w_{ij} = x_j$.
	
	In the interconnected network, \emph{adversarial} inputs and states of subsystems interact mutually, where the states of one subsystem influences an adversarial input of another (cf. interconnection constraint~\eqref{Internal}). On the contrary,  \emph{control} inputs are those not involved in constructing the interconnection: the GAS property is defined over states, while the primary aim is to design control inputs that ensure such desired property. Now, we proceed with defining an interconnected network.
	
	\begin{definition}\label{def:interconnection}
		Consider $M \in \mathbb{N}^+$ subsystems $\Sigma_i\!=\!(A_i,B_i, D_i, \R^{n_i},$ $\R^{N_i},\R^{m_i}, \R^{\sigma_i}), i \in\{1, \ldots, M\}$, with the adversarial inputs partitioned as in~\eqref{eq:partitioned w}. Let us interconnect subsystems $\Sigma_i, i \in\{1, \ldots, M\}$, with the following interconnection constraint:
		\begin{align}\label{Internal}
			w_{i j}=x_j,\quad \forall i, j \in\{1, \ldots, M\},\; i \neq j.
		\end{align} 
		Then $\Sigma=\mathcal{I}(\Sigma_1, \ldots, \Sigma_M)$ is an interconnected network of $\Sigma_i$, described by
		\begin{align}\label{eq:sys2-network}
				\Sigma\!:
				\dot x=A\mathcal F(x) + Bu,
		\end{align}
		where $A \in \mathbb{R}^{n \times N}$, with $n:=\sum_{i=1}^M n_i, N:=\sum_{i=1}^M N_i$, is an $n \times N$ block matrix with diagonal elements $(A_1,\ldots,A_M) $ and off-diagonal elements $\tilde D_i = [\tilde D_{i 1} \;\; \ldots \;\; \tilde D_{i(i-1)} \;\; \tilde D_{i(i+1)} \;\; \ldots \;\; \tilde D_{i M}]$, having values depending on the interconnection topology, where $\tilde D_{ij}  = [D_{ij} \;\; \mathbf{0}_{n_i\times(N_j - n_j)}]$, with $\mathbf{0}_{n_i\times(N_j - n_j)}$ being a zero matrix. In addition, $B = \mathsf{blkdiag}(B_1,\ldots,B_M) \in \R^{n\times m},u\in \R^{m}$ with $m:=\sum_{i=1}^M m_i$, and $\mathcal F(x) = [\mathcal F_1(x_1);\dots; \mathcal F_M(x_M)]\in \R^{N}$.  We utilize $\Sigma\!=\!(A, B, \R^{n},\R^{N},\R^{m})$ to denote the interconnected network in~\eqref{eq:sys2-network}.
	\end{definition}
	
	\subsection{GAS Property of Interconnected Network}
	
	Here, we first define the global asymptotic stability of an interconnected network~\citep{angeli2000characterization}.
	
	\begin{definition}\label{GAS}
		An interconnected network $\Sigma=\mathcal{I}(\Sigma_1, \ldots, \Sigma_M)$ achieves global asymptotic stability (GAS) if for any $x(0) \in \mathbb{R}^n$ and some function $\beta$ belonging to the class $\mathcal{KL}$, the following condition holds:
		\begin{align*}
			\Vert x(t) \Vert \leq \beta(\Vert x(0)\Vert ,t).
		\end{align*}
		This indicates that all solutions of $\Sigma$ converge to the origin, as its equilibrium point, when time goes to infinity.
	\end{definition}
	
	We now present the following theorem to demonstrate the required conditions under which an interconnected network achieves GAS according to Definition~\ref{GAS}.
	
	\begin{theorem}[\citep{sontag1996new}]
		Consider an interconnected network
		$\Sigma=\mathcal{I}(\Sigma_1, \ldots,$ $ \Sigma_M)$, composed of $M$ subsystems $\Sigma_i$. Let there exist a control Lyapunov function (CLF) $\mathcal V\!:\R^n\to\R_0^+$, and constants $ \underline{\alpha},\overline{\alpha},\kappa \in \R^+$, such that
		\begin{itemize}
			\item $\forall x\!\in\! \R^n\!\!:$
			\begin{subequations}
				\begin{equation}\label{eq:ISS-con1-network}
					\underline{\alpha}\Vert x \Vert^2\le \!\mathcal V(x) \le \overline{\alpha}\Vert x\Vert^2,
				\end{equation}
				\item $\forall x\in\! \R^n, \: \exists u\in \R^m\!\!:$
				\begin{align}\label{eq:ISS-con2-network}
					&\mathsf{L} \mathcal V(x)\mathcal ~\leq - \kappa \mathcal V(x),
				\end{align}
			\end{subequations}
		\end{itemize}
		where $x=[x_1;\dots;x_M]$ and $\mathsf{L} \mathcal V$ is the Lie derivative of $\mathcal V:\R^n\to\R_0^+$
		with respect to the dynamics in \eqref{eq:sys2-network}, defined as 
		\begin{align}
			\mathsf{L}\mathcal V(x)=\partial_{x} \mathcal V(x)(A\mathcal F(x) + B u),
		\end{align}
		with $\partial_{x} \mathcal V(x) = \frac{\partial\mathcal V(x) }{\partial_{x}}$. Then the interconnected network
		$\Sigma=\mathcal{I}(\Sigma_1,$ $ \ldots, \Sigma_M)$ is GAS in the sense of Definition~\ref{GAS}.
	\end{theorem}
	
	In general, searching for a CLF to enforce the GAS property over an interconnected network is often computationally expensive, even if the underlying model is known. To alleviate this, one approach is to analyze the stability of the network by leveraging the \emph{input-to-state} stability property of its underlying subsystems, as defined in the following definition.
	
	\begin{definition}
		Given a subsystem $\Sigma_i$, the function $\mathcal V_i:\R^{n_i}\to\R_0^+$ is
		called \emph{input-to-state stable (ISS)} Lyapunov function if there exist constants $\underline{\alpha}_i, \overline{\alpha}_i,\kappa_i \in \R^+, \rho_i \in \R^+_0,$ such that
		\begin{itemize}
			\item  $\forall x_i\!\in\R^{n_i}\!\!:$
			\begin{subequations}
				\begin{align}\label{eq:ISS-con1}
					\underline{\alpha}_i\Vert x_i\Vert^2\le \!\mathcal V_i(x_i) \le \overline{\alpha}_i\Vert x_i\Vert^2,
				\end{align}
				\item $\forall x_i\in\R^{n_i},\: \exists u_i\in \R^{m_i},$ such that $\forall w_i\in \R^{\sigma_i}\!\!:$
				\begin{align}\label{eq:ISS-con2}
					&\mathsf{L}\mathcal V_i(x_i)\mathcal ~\leq - \kappa_i \mathcal V_i(x_i) +  \rho_i\Vert w_i \Vert^2,
				\end{align}
			\end{subequations}
		\end{itemize}
		where $\mathsf{L}\mathcal V_i$ is the Lie derivative of $\mathcal V_i:\R^{n_i}\to\R_0^+$ with respect to the dynamics in \eqref{eq:subsystem}, defined as 
		\begin{align}\label{Lie derivative}
			\mathsf{L}\mathcal V_i(x_i)=\partial_{x_i} \mathcal V_i(x_i)(A_i\mathcal F_i(x_i) + B_iu_i + D_iw_i).
		\end{align}
	\end{definition}
	
	Using the ISS Lyapunov functions of individual subsystems, a CLF for the interconnected network can be constructed by satisfying a compositional condition (cf. Theorem~\ref{th:comp}), rather than treating the network as a monolithic system. However, constructing these ISS functions is infeasible due to the unknown matrices $A_i$ and $B_i$ appearing in \eqref{Lie derivative}. With this primary challenge in place, we can now formally outline the main problem of interest.
	
	\begin{resp}
		\begin{problem}\label{Problem}
			Consider an interconnected network $\Sigma =\mathcal{I}(\Sigma_1, \ldots, $ $\Sigma_M)$, composed of $M$ subsystems $\Sigma_i$, each with unknown matrices $A_i$ and $B_i$. Construct an ISS Lyapunov function $\mathcal{V}_i$ along with an ISS controller $u_i$ for each unknown subsystem by collecting solely a single input-state trajectory from it. Then, develop a compositional approach based on small-gain reasoning to compose these $\mathcal{V}_i$ and construct a control Lyapunov function $\mathcal{V}$, along with its associated controller $u$, across the interconnected network while enforcing its global asymptotic stability. 
		\end{problem}
	\end{resp}
	
	To address Problem \ref{Problem}, we introduce our data-driven scheme, in the following section, to construct ISS Lyapunov functions for unknown subsystems, derived from data.
	
	\section{Data-driven construction of ISS Lyapunov functions}\label{DD-CBCs}
	In our data-driven framework, we consider the structure of our ISS Lyapunov function to be quadratic, taking the form $\mathcal{V}_i(x_i) = x_i^\top P_i x_i$, where $P_i \succ 0$. Subsequently, we gather data from unknown subsystems over the time interval $[t_0 , t_0 + (T - 1)\tau]$, where $T \in \mathbb{N}^+$ denotes the number of collected samples, and $\tau \in \mathbb{R}^+$ represents the sampling time:
	\begin{subequations}\label{New}
		\begin{align}
			\mathcal U^{0,T}_i &= [u_i(t_0)\;\;\, u_i(t_0 + \tau)\;\;\, \dots\;\; u_i(t_0 + (T - 1)\tau)],\label{U0}\\
			\mathcal W^{0,T}_i &= [w_i(t_0)\;\; w_i(t_0 + \tau)\;\; \dots\;\; w_i(t_0 + (T - 1)\tau)],\label{W0}\\
			\mathcal X^{0,T}_i &= [x_i(t_0)\;\;\, x_i(t_0 + \tau)\;\;\, \dots\;\; x_i(t_0 + (T - 1)\tau)],\label{X0}\\
			\mathcal X^{1,T}_i &= [\dot x_i(t_0)\;\;\, \dot x_i(t_0 + \tau)\;\;\, \dots\;\; \dot x_i(t_0 + (T - 1)\tau)].\label{X1}
		\end{align}
	\end{subequations}
	We treat trajectories in~\eqref{New} as a \emph{single input-state trajectory}.
	\begin{remark}\label{remark:derivative}
	The state derivatives in $\mathcal{X}_i^{1,T}$ are not directly measurable but are essential for our framework. They can be estimated using numerical approximation based on the derivative definition or through filtering techniques~\citep{larsson2008estimation, padoan2015towards}. However, both methods introduce manageable approximation errors, which can be addressed using a strategy similar to~\cite{guo2021data}, as detailed in Subsection~\ref{sub:lim}.
	\end{remark}
	Due to the involvement of unknown matrices $A_i$ and $B_i$ in $\mathsf{L}\mathcal{V}_i(x_i)$ in \eqref{Lie derivative}, drawing inspiration from~\citep{depersis2020tac}, we provide the following lemma to derive the data-based representation for $A_i\mathcal F_i(x_i) + B_iu_i $.
	
	\begin{lemma}\label{Lemma1}
		Given a subsystem $\Sigma_i$, consider a matrix polynomial $G_i(x_i)\in\R^{T\times n_i}$ such that 
		\begin{align}\label{Q}
				\aleph_i(x_i) = \mathcal J^{0,T}_iG_i(x_i),
			\end{align}
			with $\aleph_i(x_i)\in\R^{N_i\times n_i}$ being a state-dependent transformation matrix fulfilling
			\begin{align}\label{eq:trans F}
				\mathcal F_i(x_i) = \aleph_i(x_i)x_i,
		\end{align}
		and $\mathcal J^{0,T}_i $ being a full row-rank $(N_i\times T)$ matrix, defined as
		\begin{align}\label{eq:monomial traj}
			\mathcal J^{0,T}_i \!\!= [\mathcal F_i(x_i(t_0))\;\;\mathcal F_i(x_i(t_0 + \tau))\;\;\dots\;\;\mathcal F_i(x_i(t_0 + (T - 1)\tau))].
		\end{align}
		Then, the data-based counterpart of $A_i\mathcal F_i(x_i) + B_iu_i$ is derived as 
	\begin{align}\notag
				A_i\mathcal F_i(x_i) + B_iu_i  &= (\mathcal X^{1,T}_i -D_i \mathcal W^{0,T}_i)G_i(x_i)x_i,
				\\\notag &\text{equivalently:}~\\\label{eq:data-based repr}
				A_i\aleph_i(x_i) + B_i K_i(x_i) &= (\mathcal X^{1,T}_i -D_i \mathcal W^{0,T}_i)G_i(x_i),
			\end{align}
			by designing the local controllers $u_i \!=\! K_i(x_i)x_i \!= \!\mathcal U^{0,T}_i\!G_i(x_i)x_i$.
	\end{lemma}
	
	\begin{proof}
		From trajectories in~\eqref{New}, we have
	\begin{align*} 
		\mathcal X^{1,T}_i &= A_i \mathcal J^{0,T}_i \!+ \!B_i\mathcal U^{0,T}_i \!+\! D_i\mathcal W^{0,T}_i = [B_i\quad A_i] \begin{bmatrix}
			\mathcal U^{0,T}_i\\
			\mathcal J^{0,T}_i
		\end{bmatrix} + D_i\mathcal W^{0,T}_i.
	\end{align*}
	Accordingly, 
	\begin{align} \label{new2}
		[B_i\quad A_i] \begin{bmatrix}
			\mathcal U^{0,T}_i\\
			\mathcal J^{0,T}_i
		\end{bmatrix} = \mathcal X^{1,T}_i - D_i\mathcal W^{0,T}_i.
	\end{align}
	By utilizing the local controller $u_i \!=\! K_i(x_i)x_i \! =\! \underbrace{\mathcal U^{0,T}_i \! G_i(x_i)}_{K_i(x_i)}x_i$, one has
	\begin{align}\notag
			A_i\mathcal F_i(x_i) + B_iu_i  \!&=\! A_i \aleph_i(x_i)x_i+ B_iu_i \! =\! (A_i\aleph_i(x_i) + B_iK_i(x_i))x_i\\\notag
			&= [B_i\quad A_i] \begin{bmatrix}
				K_i(x_i)\\
				\mathds \aleph_i(x_i)
			\end{bmatrix}x_i \\\label{new3}
			&= [B_i\quad A_i] \begin{bmatrix}
				\mathcal U^{0,T}_i\\
				\mathcal J^{0,T}_i
			\end{bmatrix}G_i(x_i) x_i.
	\end{align}
	The last equality in~\eqref{new3} is established due to~\eqref{Q} and since $K_i(x_i) =  \mathcal U^{0,T}_iG_i(x_i)$. Now by applying~\eqref{new2} to~\eqref{new3}, we have 
	\begin{align}\notag
			A_i\mathcal F_i(x_i) + B_iu_i &= (\mathcal X^{1,T}_i -D_i \mathcal W^{0,T}_i)G_i(x_i)x_i, \\\notag
			&\text{equivalently:}~\\\notag
			A_i \aleph_i(x_i)+ B_i K_i(x_i) &= (\mathcal X^{1,T}_i -D_i \mathcal W^{0,T}_i)G_i(x_i),
	\end{align}
	which concludes the proof. 
	\end{proof}
	
	\begin{remark}
		To guarantee that $\mathcal{J}^{0,T}_i$ in~\eqref{eq:monomial traj} achieves full row rank, the number of samples $T$ should not be less than or equal to $N_i$. This assumption is required to ensure that~\eqref{Q} is feasible and can be readily confirmed since the matrix $\mathcal{J}^{0,T}_i$ is derived from sampled data. This assumption also complies with Willem et al's fundamental lemma~\cite{willems2005note} regarding the persistently exciting data.
	\end{remark}
	
	\begin{remark}Note that the full row-rank condition on~\eqref{eq:monomial traj} is necessary and sufficient for the feasibility of~\eqref{Q}. To substantiate this claim, we define $G_i(x_i):=\Upsilon_{1i}\Upsilon_{2i}(x_i)$, where $\Upsilon_{1i}\in\mathbb{R}^{T\times N_i}$ and $\Upsilon_{2i}(x_i)\in\mathbb{R}^{N_i\times n_i}$. Since $\mathcal{J}_i^{0,T}$ is assumed to be a full row-rank matrix, its right pseudo-inverse $(\mathcal{J}_i^{0,T})^\ddagger$ always exists. Thus, we can set $\Upsilon_{1i}=(\mathcal{J}_i^{0,T})^\ddagger$ and $\Upsilon_{2i}(x_i)=\aleph_i(x_i)$. Consequently, this gives $G_i(x_i)=\Upsilon_{1i}\Upsilon_{2i}(x_i)=(\mathcal{J}_i^{0,T})^\ddagger\aleph_i(x_i)$, leading to $\mathcal{J}_i^{0,T}G_i(x_i)=\mathcal{J}_i^{0,T}(\mathcal{J}_i^{0,T})^\ddagger\aleph_i(x_i)=\mathds{I}_{N_i}\aleph_i(x_i)=\aleph_i(x_i)$, which confirms the existence of a matrix $G_i(x_i)$ satisfying~\eqref{Q} when $\mathcal{J}_i^{0,T}$ is full row-rank.
	\end{remark}
	
	By leveraging the data-based representation of $A_i \aleph_i(x_i)+ B_i K_i(x_i)$ as described in Lemma~\ref{Lemma1}, we introduce the following theorem, as one of the central contributions of our work,  to construct ISS Lyapunov functions alongside their corresponding local controllers directly from data for unknown subsystems. 
	
	\begin{theorem}\label{Thm:main3}
		Given a ct-NPS $\Sigma_i$ in \eqref{eq:subsystem}, with matrices $A_i$ and $B_i$ being unknown, consider the data-based representation~\eqref{eq:data-based repr}
		as in Lemma~\ref{Lemma1}. Let  there exist $P_i\in\R^{n_i\times n_i}$ and a matrix polynomial
		$\Phi_i(x_i) \in \R^{T \times n_i}$ such that
		\begin{subequations}
			\begin{align}\label{eq:con1 SOS}
				& \mathcal J_i^{0,T}\Phi_i(x_i) = \aleph_i(x_i)P_i^{-1}, \quad \text{with} ~ P_i \succ 0.
			\end{align}
			If 
			\begin{align}\label{eq:con2 SOS}
				\mathcal C \preceq -\kappa_i P_i^{-1},
			\end{align}
		\end{subequations}
		where
		\begin{align*}
			\mathcal C =  ~\!&(\mathcal X^{1,T}_i - D_i\mathcal W^{0,T}_i) \Phi_i(x_i) + \Phi_i(x_i)^\top(\mathcal X^{1,T}_i - D_i\mathcal W^{0,T}_i)^\top + \pi_i \mathds I_{n_i},
		\end{align*}
		for some $\kappa_i, \pi_i\in \R^+$, then $\mathcal V_i(x_i) = x_i^\top [\aleph_i^\dagger\mathcal J_i^{0,T}\Phi_i(x_i)]^{-1}x_i$ is an ISS Lyapunov function for $\Sigma_i$ with $\underline{\alpha}_i = \lambda_{\min}(P_i)$ and $\overline{\alpha}_i = \lambda_{\max}(P_i)$, $ \rho_i = \frac{\Vert D_i\Vert^2}{\pi_i}$, and $u_i = \mathcal U^{0,T}_i\Phi_i(x_i)[\aleph_i^\dagger\mathcal J_i^{0,T}\Phi_i(x_i)]^{-1}x_i$ is its local ISS controller.
	\end{theorem}

	\begin{proof}
		We first proceed with showing condition~\eqref{eq:ISS-con1}. Since
	\begin{align*}
		\lambda_{\min}(P_i)\Vert x_i\Vert^2 &\leq \underbrace{x_i^\top \overbrace{ [\aleph_i^\dagger\mathcal J_i^{0,T}\Phi_i(x_i)]^{-1}}^{P_i} x_i}_{\mathcal V_i(x_i)} \leq \lambda_{\max}(P_i)\Vert x_i \Vert ^2,
	\end{align*}
	one can readily satisfy condition~\eqref{eq:ISS-con1} by choosing $\underline{\alpha}_i = \lambda_{\min}(P_i)$ and $\overline{\alpha}_i = \lambda_{\max}(P_i)$.
	
	We now move forward by demonstrating condition~\eqref{eq:ISS-con2}, as well. Since $\aleph_i(x_i)P_i^{-1}\!\! = \!\! \mathcal J_i^{0,T}\Phi_i(x_i)$ according to~\eqref{eq:con1 SOS}, then $\aleph_i(x_i) \!=\! \mathcal J_i^{0,T}\Phi_i(x_i) P_i $. As $\aleph_i(x_i) \!=\! \mathcal J_i^{0,T}G_i(x_i)$ according to~\eqref{Q}, then $G_i(x_i)$ $= \Phi_i(x_i)P_i$ can be a proper choice, and accordingly, $G_i(x_i) P_i^{-1} = \Phi_i(x_i)$. Given that $A_i \aleph_i(x_i)+ B_iK_i(x_i) = (\mathcal X^{1,T}_i - D_i\mathcal W^{0,T}_i)G_i(x_i)$ as per Lemma~\ref{Lemma1}, we have
	\begin{align}\label{new11}
		(A_i \aleph_i(x_i)+ B_iK_i(x_i))P_i^{-1} &= (\mathcal X^{1,T}_i - D_i\mathcal W^{0,T}_i)\underbrace{G_i(x_i)P_i^{-1}}_{\Phi_i(x_i)} \notag\\
		&=(\mathcal X^{1,T}_i - D_i\mathcal W^{0,T}_i) \Phi_i(x_i).
	\end{align}
	By leveraging the definition of Lie derivative in \eqref{Lie derivative}, and since $P_i = [\aleph_i^\dagger\mathcal J_i^{0,T}\Phi_i(x_i)]^{-1}$ as per~\eqref{eq:con1 SOS}, we have
	\begin{align*}
		\mathsf{L}\mathcal V_i(x_i)=&~\partial_{x_i} \mathcal V_i(x_i)(A_i\mathcal F_i(x_i) + B_i u_i + D_iw_i)\\
		=&~2x_i^\top P_i(A_i\mathcal F_i(x_i) + B_i u_i + D_iw_i) \\
		=&~2x_i^\top P_i\big((A_i\aleph_i(x_i)+B_i K_i(x_i))x_i + D_iw_i\big)\\
		=&~2x_i^\top P_i(A_i\aleph_i(x_i)+B_i K_i(x_i))x_i +2 \underbrace{x_i^\top P_i}_{a_i}\underbrace{D_iw_i}_{b_i}.
	\end{align*}
	Using Cauchy-Schwarz inequality~\citep{bhatia1995cauchy} as ${a_i} b_i \leq \Vert a_i \Vert \Vert b_i\Vert,$ for any ${a_i^\top},b_i\in \mathbb R^n$, and by leveraging Young's inequality~\citep{young1912classes} as $\Vert a_i\Vert \Vert b_i\Vert\leq \frac{\pi_i}{2}\Vert a_i\Vert^2+\frac{1}{2\pi_i}\Vert b_i\Vert^2,$ for any $\pi_i>0$, one can attain
	\begin{align*}
		\mathsf{L}\mathcal V_i(x_i)\leq &~2x_i^\top P_i(A_i\aleph_i(x_i)+B_i K_i(x_i))x_i \\
		&+ \pi_i x_i^\top P_i  P_ix_i+ \frac{1}{\pi_i}\Vert D_iw_i\Vert^2.
	\end{align*}
	By expanding the above expression and factorizing the term $x_i^\top P_i$ from left and $P_i x_i$ from right, and utilizing Cauchy-Schwarz inequality for the last term, we have
	\begin{align*}
		\mathsf{L}\mathcal V_i(x_i)\leq&~x_i^\top P_i\Big[(A_i\aleph_i(x_i)+B_i K_i(x_i))P_i^{-1}\\
		&+ P_i^{-1} (A_i\aleph_i(x_i)+B_i K_i(x_i))^\top + \pi_i \mathds I_{n_i}\Big] P_i x_i\\
		& + \frac{1}{\pi_i}\Vert D_i\Vert^2\Vert w_i \Vert^2.
	\end{align*}
	Since $(A_i \aleph_i(x_i)+ B_iK_i(x_i))P_i^{-1} = \big(\mathcal X^{1,T}_i - D_i\mathcal W^{0,T}_i\big) \Phi_i(x_i)$ as per \eqref{new11}, one can conclude that
	\begin{align*}
		\mathsf{L}\mathcal V_i(x_i) \leq& ~x_i^\top \!P_i\Big[\big(\mathcal X^{1,T}_i \!-\! D_i\mathcal W^{0,T}_i\big) \Phi_i(x_i)\!\\
		&+ \Phi_i(x_i)^\top\!\big(\mathcal X^{1,T}_i \!-\! D_i\mathcal W^{0,T}_i\big)^\top \!+ \pi_i \mathds I_{n_i} \Big] P_i x_i \\
		&+ \frac{1}{\pi_i}\Vert D_i\Vert^2\Vert w_i \Vert^2.
	\end{align*}
	Now, by considering condition \eqref{eq:con2 SOS}, one has
	\begin{align*}
		\mathsf{L}\mathcal V_i(x_i) \leq &  - \kappa_i x_i^\top P_i\overbrace{P_i^{-1} P_i}^{\mathds I_{n_i}} x_i+ \frac{1}{\pi_i}\Vert D_i\Vert^2\Vert w_i \Vert^2 \\
		= &- \kappa_i \mathcal V_i(x_i) +  \rho_i\Vert w_i\Vert^2,
	\end{align*}
	satisfying condition~\eqref{eq:ISS-con2} with $ \rho_i = \frac{\Vert D_i\Vert^2}{\pi_i}$.  Then one can conclude that $\mathcal V_i(x_i) = x_i^\top [\aleph_i^\dagger\mathcal J_i^{0,T}\Phi_i(x_i)]^{-1}x_i$ is an ISS Lyapunov function for $\Sigma_i$ and $u_i = \mathcal U^{0,T}_i\Phi_i(x_i)[\aleph_i^\dagger\mathcal J_i^{0,T}\Phi_i(x_i)]^{-1}x_i$ is its local ISS controller, thereby completing the proof.
	\end{proof}
	\begin{remark}
			Of note is that condition~\eqref{eq:trans F} is essential to our approach, as it enables everything to be expressed in terms of $x_i$ rather than $\mathcal F_i(x_i)$. This aligns with the structure of our ISS Lyapunov function, $\mathcal V_i(x_i) = x_i^\top P_i x_i $, and aids in proving our main results in Theorem~\ref{Thm:main3}. Given that $\mathcal F_i(0) = 0$, there always exists a transformation $\aleph_i(x_i)$ that meets condition~\eqref{eq:trans F}, enabling us to define $\mathcal F_i(x_i)$ accordingly.
	\end{remark}
	\begin{remark}\label{SOS}
		In order to enforce conditions~\eqref{eq:con1 SOS}-\eqref{eq:con2 SOS}, one can utilize existing software tools in the literature such as \textsf{SOSTOOLS} \citep{papachristodoulou2013sostools}  along with a semidefinite programming (SDP) solver such as \textsf{SeDuMi}~\citep{sturm1999using}.
	\end{remark}

	In the next section, we propose a compositional condition using \emph{small-gain reasoning} to construct a control Lyapunov function for the network along with its controller, leveraging ISS Lyapunov functions of its individual subsystems derived from data, thereby ensuring the network's satisfaction of the GAS property.
	
	\section{Compositional construction of CLF for interconnected network}
	We analyze the network $\Sigma =\mathcal{I}(\Sigma_1, \ldots, \Sigma_M)$ by deriving a small-gain type compositional condition, exploring the construction of a CLF for the network based on data-driven ISS Lyapunov functions of its subsystems. To do so, we first define $\mathcal K:=\mathsf{diag}(\kappa_1, \dots, \kappa_M)$ and $\hat \varrho:=\{\hat{\rho}_{i j}\}$ with $\hat{\rho}_{i j}=\frac{\rho_i}{\underline{\alpha}_j}$, where $\hat{\rho}_{i i}=0,$ for any $i \in\{1, \dots, M\}$.
	
	We offer the following theorem to demonstrate the conditions under which one can construct a CLF for the unknown network using data-driven ISS Lyapunov functions of the unknown subsystems.
	
	\begin{theorem}\label{th:comp}
		Consider a network $\Sigma=\mathcal{I}(\Sigma_1, \dots, \Sigma_M)$, comprised of $M \in \mathbb{N}^{+}$ subsystems $\Sigma_i$. Let each $\Sigma_i$ possess an ISS Lyapunov function $\mathcal V_i$ constructed from data, according to Theorem \ref{Thm:main3}. If
		\begin{align}\label{eq:comp con}
			&\mathds{1}_M^\top(-\mathcal K+\hat \varrho):=[\mu_1 ; \dots ; \mu_M]^{\top}<0 \\
			&\text{equivalently, } \quad \mu_i<0, ~\forall i \in\{1, \dots, M\},\notag
		\end{align}
		then
		\begin{align*}
			\mathcal V(x):=\sum_{i=1}^M \mathcal V_i( x_i) = \sum_{i=1}^M x_i^\top [\aleph_i^\dagger\mathcal J_i^{0,T}\Phi_i(x_i)]^{-1}x_i
		\end{align*}
		is a control Lyapunov function for network $\Sigma$ with
		\begin{align*}
			\kappa:=-\mu,\max _{1 \leq i \leq M} \mu_i<\mu<0,\quad\underline{\alpha}\Let \min_{i}\{\underline{\alpha}_i\},\quad \overline{\alpha}\Let \max_i\{\overline{\alpha}_i\}.
		\end{align*}
		Moreover, $ u \!=\! [u_1;\dots; u_M]$ with
		\begin{align*}
			u_i \!=\! \mathcal U^{0,T}_i\Phi_i(x_i) [\aleph_i^\dagger\mathcal J_i^{0,T}\!\Phi_i(x_i)]^{-1}\!x_i,\;\; i\in\{1,\dots, M\},
		\end{align*}
		is a controller that renders the network GAS. 
	\end{theorem}
	\begin{proof}
		We first show that condition~\eqref{eq:ISS-con1-network} is satisfied. According to condition~\eqref{eq:ISS-con1}, one has
	\begin{align*}
		\Vert x \Vert^2 = \sum_{i = 1}^{M}\Vert x_i \Vert^2 \leq \sum_{i = 1}^{M} \frac{1}{\underline{\alpha}_i}\mathcal V_i(x_i) \leq \eta \sum_{i = 1}^{M} \mathcal V_i(x_i) = \eta \mathcal V(x),
	\end{align*}
	where $\eta = \underset{i}{\max}\{\frac{1}{\underline{\alpha}_i}\}$. Then by choosing $\underline{\alpha} = \frac{1}{\eta} = \underset{i}{\min}\{\underline{\alpha}_i\}$, we have
	\begin{align*}
		\underline{\alpha}\Vert x \Vert^2 \leq \mathcal V(x).
	\end{align*}
	On the other side, we have
	\begin{align*}
		\mathcal V(x) \!=\! \sum_{i = 1}^{M} \mathcal V_i(x_i) \!\leq \!\sum_{i = 1}^{M} \!\overline{\alpha}_i \Vert x_i \Vert^2 \leq \overline{\alpha}\! \sum_{i = 1}^{M} \!\Vert x_i \Vert^2 \!= \!\overline{\alpha} \Vert x \Vert^2,
	\end{align*}
	with $\overline{\alpha}=\underset{i}{\max}\{\overline{\alpha}_i\}$, implying that condition~\eqref{eq:ISS-con1-network} is fulfilled with $\underline{\alpha} = \underset{i}{\min}\{\underline{\alpha}_i\}$ and $\overline{\alpha}=\underset{i}{\max}\{\overline{\alpha}_i\}$.
	
	We now proceed with demonstrating that condition~\eqref{eq:ISS-con2-network} is met, as well. By employing condition~\eqref{eq:ISS-con2}, compositional condition $\mathds 1_M^{\top}(-\mathcal K+\hat \varrho)<0$, and by defining
	\begin{align*}
		-\kappa s & :=\max \big\{\mathds 1_M^{\top}(-\mathcal K+\hat \varrho) \overline{\mathcal V}(x) \mid \mathds 1_M^{\top} \overline{\mathcal V}(x)=s\big\},
	\end{align*}
	where $\overline{\mathcal V}(x)=[\mathcal V_1( x_1) ; \dots ; \mathcal V_M(x_M)]$, one can obtain the following chain of inequalities:
	\begin{align*}
		&\mathsf{L}\mathcal V(x) \\
		&= \mathsf{L}\sum_{i=1}^{M} \mathcal V_i(x_i) = \sum_{i=1}^{M} \mathsf{L}\mathcal V_i(x_i) \leq \!\!\sum_{i=1}^{M}\big(-\kappa_i \mathcal V_i(x_i)+\rho_i\Vert w_i\Vert^2\big)\\
		&= \!\!\sum_{i=1}^{M}\!\big(\!-\kappa_i \mathcal V_i(x_i)\!+\!\!\sum_{\substack{j=1 \\ j\neq i}}^{M}\!\rho_i\Vert w_{ij}\Vert^2\big) =\!\! \sum_{i=1}^{M}\!\big(\!-\kappa_i \mathcal V_i(x_i)+\!\!\sum_{\substack{j=1 \\ j\neq i}}^{M}\!\rho_i\Vert x_j\Vert^2\big)\\
		&\leq \!\!\sum_{i=1}^{M}\!\big(\!-\kappa_i \mathcal V_i(x_i)+\!\!\sum_{\substack{j=1 \\ j\neq i}}^{M}\!\frac{\rho_i}{\underline{\alpha}_j}\!\mathcal V_j(x_j)\big)~\!=\!\!\sum_{i=1}^{M}\!\big(\!-\kappa_i \mathcal V_i(x_i)+\!\!\sum_{\substack{j=1 \\ j\neq i}}^{M}\!\hat \rho_{i j}\!\mathcal V_j(x_j)\big)\\
		&= \mathds 1_M^\top(-\mathcal K + \hat \varrho)\underbrace{[\mathcal V_1(x_1); \dots; \mathcal V_M(x_M)]}_{\overline{\mathcal V}(x)}\leq-\kappa \mathcal V(x).
	\end{align*}
	We now show that $\kappa>0$ by considering $\kappa = -\mu$. Since $\mathds 1_M^{\top}(-\mathcal K+\hat \varrho):=[\mu_1 ; \dots ; \mu_M]^{\top}<0$, and $\underset{1 \leq i \leq M}{\max}\mu_i<\mu<0$, one has
	\begin{align*}
		-\kappa s&=\mathds 1_M^{\top}(-\mathcal K+\hat \varrho) \overline{\mathcal V}(x)=[\mu_1 ; \dots ; \mu_M]^{\top}[\mathcal V_1(x_1);\dots ; \mathcal V_M(x_M)]\\
		&=\mu_1 \mathcal V_1(x_1)+\dots+\mu_M\mathcal V_M(x_M) \leq \mu\big(\mathcal V_1(x_1)+\dots+\mathcal V_M(x_M)\big)\\
		&=\mu s.
	\end{align*}
	Then, $-\kappa s \leq \mu s$, and accordingly, $-\kappa \leq \mu$ as $s$ is positive. Since $\underset{1 \leq i \leq M}{\max}\mu_i<\mu<0$, then $\kappa=-\mu > 0$, implying that condition~\eqref{eq:ISS-con2-network} is fulfilled, which completes the proof.
	\end{proof}
	
	We provide Algorithm~\ref{Alg:1} as a summary of our \emph{compositional data-driven} approach for constructing a CLF and its GAS controller for an unknown network, relying on ISS Lyapunov functions of its individual subsystems.
	
	\section{Discussions}
	\begin{algorithm}[t!]
		\caption{Data-driven construction of CLF and its GAS controller}\label{Alg:1}
		\begin{center}
			\begin{algorithmic}[1]
				\REQUIRE 
				A choice for $\mathcal F_i(x_i)$
				\FOR{$i\in\{1, \dots, M\}$ over the time interval $[t_0, t_0 + (T-1)\tau]$,}\label{line 1}
				\STATE
				Collect data $\mathcal U_i^{0,T}$, $\mathcal W_i^{0,T}$, $\mathcal X_i^{0,T}$, and $\mathcal X_i^{1,T}$ using \eqref{U0}-\eqref{X1}
				\STATE
				Compute $\mathcal J^{0,T}_i$ as in~\eqref{eq:monomial traj}
				\ENDFOR
				\FOR{$i\in\{1, \dots, M\}$,}
				\STATE
				Utilize \textsf{SOSTOOLS} and design $\Phi_i(x_i)$ through fulfillment of conditions $\mathcal J_i^{0,T}\Phi_i(x_i) = \aleph_i(x_i)\Xi$\footnotemark    in~\eqref{eq:con1 SOS}, and $\mathcal C \preceq -\kappa_i \Xi$ in~\eqref{eq:con2 SOS}, with $\Xi \succ 0$, for fixed $\kappa_i, \pi_i \in \R^+$
				\STATE
				Construct ISS Lyapunov function $\mathcal V_i(x_i) = x_i^\top [\aleph_i^\dagger\mathcal J_i^{0,T}\Phi_i(x_i)]^{-1}x_i = x_i^\top \Xi^{-1}x_i$ and controller $u_i = \mathcal U^{0,T}_i\Phi_i(x_i)[\aleph_i^\dagger\mathcal J_i^{0,T}\Phi_i(x_i)]^{-1}x_i = \mathcal U^{0,T}_i \Phi_i(x_i)\Xi^{-1}x_i $
				\ENDFOR
				\IF{compositional condition~\eqref{eq:comp con} holds}
				\STATE
				$\mathcal V(x) \Let \sum_{i = 1}^{M} \mathcal V_i(x_i)$ is a CLF for the network $\Sigma$, and controller $u = [u_1;\dots;u_M]$ makes the network GAS
				\ELSE
				\STATE
				Return to Step \ref{line 1} and repeat the required procedures with a larger dataset comprising more collected samples $T$
				\ENDIF
				\ENSURE CLF $\mathcal V(x)= \sum_{i=1}^M x_i^\top [\aleph_i^\dagger\mathcal J_i^{0,T}\Phi_i(x_i)]^{-1}x_i$, controller $u = [u_1;\dots;u_M]$ with $	u_i \!=\! \mathcal U^{0,T}_i\Phi_i(x_i) [\aleph_i^\dagger\mathcal J_i^{0,T}\!\Phi_i(x_i)]^{-1}\!x_i,$ satisfaction of GAS property across the network
			\end{algorithmic}
		\end{center}
	\end{algorithm}
	
	\subsection{Computational Complexity Analysis}
	One can leverage SOS optimization programs to impose conditions~\eqref{eq:con1 SOS}-\eqref{eq:con2 SOS}, as illustrated in Remark~\ref{SOS}. The computational complexity of SOS is contingent upon both the polynomial degree of candidate Lyapunov functions and the dimension of state variables. Existing research demonstrates that with a fixed degree of Lyapunov functions, computational demands for designing a CLF escalates \emph{polynomially} as the network dimension expands (see, for instance, \citep{wongpiromsarn2015automata} for a similar argument for barrier functions). In contrast, our compositional data-driven approach substantially mitigates this complexity, scaling \emph{linearly} with the number of subsystems. It is worth mentioning that in our work, the complexity still scales polynomially within each \emph{subsystem level}, but linearly with the number of subsystems, thus making our approach practical for analyzing networks with a large number of subsystems. In particular, the polynomial complexity at the subsystem level is not restrictive in our setting, as we applied our findings to networks with $2000$ dimensions, each comprising subsystems of 3 dimensions (cf. Table \ref{tab:benchmark}). This expands the applicability of formal data-driven techniques to large-scale networks.
	
	\subsection{Limitations}\label{sub:lim}
	\footnotetext{By defining $\Xi = P^{-1}$, we address conditions \eqref{eq:con1 SOS} and \eqref{eq:con2 SOS} simultaneously so that the product of matrices $\mathcal J_i^{0,T}$ and $\Phi_i(x_i)$ yields the product of $\aleph_i(x_i)$ and a symmetric, positive-definite matrix $\Xi$, and $\mathcal{C}\preceq-\kappa_i\Xi$. Subsequently, $\Xi$ equates to the inverse of matrix $P$, and upon inversion, we obtain matrix $P$ itself.}
	While our proposed approach offers numerous benefits, it also has certain limitations, as is the case with any method, that are worth noting. First, as previously mentioned, we allow $\mathcal F_i(x_i)$ to include all monomials up to a maximum degree informed by physical insights into the unknown system. This means $\mathcal F_i(x_i)$ may contain extra terms that do not exist in the actual dynamics, potentially making the approach conservative. While the inclusion of these redundant terms can increase computational costs, it ensures that $\mathcal F_i(x_i)$ potentially includes all monomial terms present in the actual unknown subsystem, enabling one to provide formal GAS certificates for networks with unknown subsystems. Moreover, this work focuses on polynomial nonlinearities, whereas practical system dynamics may involve broader nonlinear terms, such as fractional, exponential, or trigonometric functions. Extending our results to a wider class of nonlinear systems beyond polynomials is an ongoing research direction.
	
		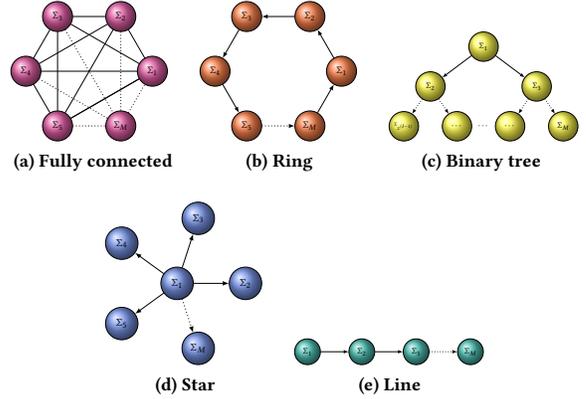
\begin{figure}[t!]
		\centering
		\subfloat[\centering Fully connected\label{fig:fully}]{{\resizebox{0.25\linewidth}{!}{\begin{tikzpicture}
						\tikzset{vertex/.style = {
								shape=circle,
								draw,
								ball color=vividcerise!70, 
								minimum size=3em
						}}
						\tikzset{edge/.style = {draw, thick}}
						\node[vertex] (a) at (0:2) {$\Sigma_1$};
						\node[vertex] (b) at (60:2) {$\Sigma_2$};
						\node[vertex] (c) at (120:2) {$\Sigma_3$};
						\node[vertex] (d) at (180:2) {$\Sigma_4$};
						\node[vertex] (e) at (240:2) {$\Sigma_5$};
						\node[vertex] (f) at (300:2) {$\Sigma_M$};
						\draw[edge] (a) -- (b);
						\draw[edge] (b) -- (c);
						\draw[edge] (c) -- (d);
						\draw[edge] (d) -- (e);
						\draw[edge, dotted] (e) -- (f);
						\draw[edge, dotted] (f) -- (a);
						\draw[edge] (a) -- (c);
						\draw[edge] (a) -- (d);
						\draw[edge] (a) -- (e);
						\draw[edge] (b) -- (d);
						\draw[edge] (b) -- (e);
						\draw[edge, dotted] (b) -- (f);
						\draw[edge] (c) -- (e);
						\draw[edge, dotted] (c) -- (f);
						\draw[edge, dotted] (d) -- (f);
						\draw[edge ] (e) -- (a);
					\end{tikzpicture}
		} }}%
		\quad
		\subfloat[\centering Ring\label{fig:ring}]{\resizebox{0.25\linewidth}{!}{%
				\begin{tikzpicture}
					\tikzset{vertex/.style = {
							shape=circle,
							draw,
							ball color=tangelo!70, 
							minimum size=3em
					}}
					\tikzset{edge/.style = {draw, thick, -latex}}
					\node[vertex] (1) at ({360/6 * (1 - 1)}:2cm) {$\Sigma_1$};
					\node[vertex] (2) at ({360/6 * (2 - 1)}:2cm) {$\Sigma_2$};
					\node[vertex] (3) at ({360/6 * (3 - 1)}:2cm) {$\Sigma_3$};
					\node[vertex] (4) at ({360/6 * (4 - 1)}:2cm) {$\Sigma_4$};
					\node[vertex] (5) at ({360/6 * (5 - 1)}:2cm) {$\Sigma_5$};
					\node[vertex] (6) at ({360/6 * (6 - 1)}:2cm) {$\Sigma_M$};
					\draw[edge] (1) -- (2);
					\draw[edge] (2) -- (3);
					\draw[edge] (3) -- (4);
					\draw[edge] (4) -- (5);
					\draw[edge, dotted] (5) -- (6);
					\draw[edge] (6) -- (1);
				\end{tikzpicture}
		}}
		\quad
		\subfloat[\centering Binary tree\label{fig:binary}]{
			\resizebox{0.3\linewidth}{!}{%
				\begin{tikzpicture}[level/.style={sibling distance = 4cm/#1,
						level distance = 1.5cm}]
					\tikzset{vertex/.style = {shape=circle,draw,ball color= lemon!70, minimum size=3.5em}}
					\tikzset{edge/.style = {-latex, draw, thick}}  
					\tikzset{dashdot/.style = {draw, thick, dotted, -latex}}
					\node [vertex] (a) {$\Sigma_1$}
					child {node [vertex] (b) {$\Sigma_2$}[-latex, thick]
						child {node [vertex,thin] (d) {{\tiny $\Sigma_{2^{(l-1)}}$}}[dotted]
						}
						child {node [vertex,thin] (e) {$\cdots$}[dotted]
						}
					}
					child {node [vertex] (c) {$\Sigma_3$}[-latex, thick]
						child {node [vertex,thin] (f) {$\cdots$}[dotted]
						}
						child {node [vertex,thin] (g) {$\Sigma_M$}[dotted]
						}
					};
					\path (e.east) -- (f.west) node [midway] {$\cdots$};
				\end{tikzpicture}
		}}
		\\
		\subfloat[\centering Star \label{fig:star}]{{\resizebox{0.25\linewidth}{!}{%
					\begin{tikzpicture}
						\tikzset{vertex/.style = {
								shape=circle,
								draw,
								ball color= royalblue(web)!70, 
								minimum size=3em
						}}
						\tikzset{edge/.style = {draw, thick, -latex}}
						\node[vertex] (center) at (0,0) {$\Sigma_1$};
						\node[vertex] (1) at (0:2) {$\Sigma_2$};
						\node[vertex] (2) at (72:2) {$\Sigma_3$};
						\node[vertex] (3) at (144:2) {$\Sigma_4$};
						\node[vertex] (4) at (216:2) {$\Sigma_5$};
						\node[vertex] (5) at (288:2) {$\Sigma_M$};
						\draw[edge] (center) -- (1);
						\draw[edge] (center) -- (2);
						\draw[edge] (center) -- (3);
						\draw[edge] (center) -- (4);
						\draw[edge, dotted] (center) -- (5);
					\end{tikzpicture}
		} }}%
		\quad
		\subfloat[\centering Line\label{fig:line}]{\resizebox{0.3\linewidth}{!}{%
				\begin{tikzpicture}[node distance=2cm]
					\tikzset{vertex/.style = {
							shape=circle,
							draw,
							ball color= persiangreen!70, 
							minimum size=3em
					}}
					\tikzset{edge/.style = {draw, thick}}
					\node[vertex] (1) {$\Sigma_1$};
					\node[vertex] (2) [right of=1] {$\Sigma_2$};
					\node[vertex] (3) [right of=2] {$\Sigma_3$};
					\node[vertex] (4) [right of=3] {$\Sigma_M$};
					\draw[edge,-latex,thick] (1) -- (2);
					\draw[edge,-latex,thick] (2) -- (3);
					\draw[edge, dotted,-latex,thick] (3) -- (4);
				\end{tikzpicture}
		}}
		\caption{Different interconnection topologies employed in Table \ref{tab:benchmark}. In the binary tree topology shown in~ \protect\subref{fig:binary}, $M=2^l-1$, where $l$ is the number of tree's levels.\label{fig:topologies}}%
	\end{figure}
	
Another point of note, as briefly mentioned in Remark~\ref{remark:derivative}, is that the derivative of states encoded in $\mathcal X_i^{1,T}$ cannot be measured directly, whereas their data is required in our framework. Two immediate solutions to overcome this burden are either to obtain their values numerically using $\frac{x_i(t_0+\tau) - x_i(t_0)}{\tau}$ or to employ suitable filters to estimate them via existing approaches~\citep{larsson2008estimation, padoan2015towards}. Nevertheless, both methods can introduce approximation errors, which are manageable and not a significant concern. Particularly, a technique similar to that in~\cite{guo2021data} can be employed to integrate the impact of error in the derivative approximation. More precisely, in this approach, the data $\mathcal X^{1,T}_i$ is expressed as $\mathcal X^{1,T}_i = \widetilde{\mathcal X}^{1,T}_i + \Delta_i$, where $\widetilde{\mathcal X}^{1,T}_i$ represents the error-free data and $\Delta_i$ is a small noise capturing the error. The only assumption required is that $\Delta_i$ satisfies $\Delta_i\Delta_i^\top \preceq \Lambda_i\Lambda_i^\top$ for some known matrix $\Lambda_i$, which is a common assumption in control and intuitively indicates that the noise energy during data collection is bounded. We opted not to use this method in the paper for the sake of a clearer presentation and since it does not represent a core contribution of this work.
	
	\section{Simulation results}\label{compu}
	In this section, we demonstrate the effectiveness of our compositional data-driven approach by applying it to three networks: two real-world systems, including spacecraft~\citep{khalil2002control} and \emph{chaotic} Lorenz-type systems (referred to as Lu) \citep{strogatz2018nonlinear}, and an academic example, all \emph{interconnected via various topologies}, as shown in Figure~\ref{fig:topologies}. A general yet brief overview of these case studies is provided in Table \ref{tab:benchmark}. In all case studies, we assume that the matrices $A_i$, $B_i$ of subsystems and, accordingly, $A$ and $B$ of the networks are unknown. All simulations are performed in \MATLAB \textsl{R2023b} on a MacBook Pro (Apple M2 Pro with 16GB memory). 
	
	The primary goal across all case studies is to design a CLF and its GAS controller for the interconnected network with an unknown mathematical model. To accomplish this, under Algorithm~\ref{Alg:1}, we  gather input-state trajectories from subsystems as described in~\eqref{New} and satisfy condition~\eqref{eq:con2 SOS} by constructing ISS Lyapunov functions and local controllers for all subsystems. Then, utilizing the compositionality results from Theorem~\ref{th:comp}, we compositionally design a CLF and its controller for the interconnected network based on ISS Lyapunov functions of subsystems, ensuring the network is GAS. In the subsequent subsections, each case study is presented and the corresponding results are reported.
	
	\begin{table*}[h!]
		\centering
		\caption{An overview of our data-driven findings across a set of interconnected networks, including runtime (\texttt{RT}) and memory usage (\texttt{MU}), required for each subsystem. $M$ represents the number of subsystems, and $T$ denotes the number of collected samples.\label{tab:benchmark}}
		{\small \begin{tabular}{@{}ccccccc@{}}
			\toprule
			\multirow{2}{*}[-0.25em]{System} & \multirow{2}{*}[-0.25em]{Selected $\mathcal F_i(x_i)$} & \multicolumn{3}{c}{Network parameters} & \multicolumn{2}{c}{Computation costs}\\
			\cmidrule(lr){3-5} \cmidrule(lr){6-7} \vspace{-0.25cm}\\
			{} & {} & Topology & $M$ & $T$  & \texttt{RT} (sec) & \texttt{MU} (Mb)\\
			\midrule
			\myalign{l}{\multirow{2}{*}{Spacecraft}} & \multirow{2}{*}{$[x_{1i};x_{2i};x_{3i};x_{1i}^2;x_{2i}^2;x_{3i}^2;x_{1i}x_{2i};x_{1i}x_{3i};x_{2i}x_{3i}]$} & \myalign{l}{Fully connected}  & $1000$ & $100$  & $< 1$ & $3.54$ \\
			{} & {} & \myalign{l}{Ring}   & $2000$  & $15$ & $< 1$ & $< 1$\\
			\midrule
			\myalign{l}{\multirow{2}{*}{Academic network}} & \multirow{2}{*}{$[x_{1i};x_{2i};x_{1i}^2;x_{2i}^2;x_{1i}x_{2i}]$}  &  \myalign{l}{Binary tree}   &	$2047$	&	$12	$		&	$< 1$	& $< 1$\\
			{}	& {} &	\myalign{l}{Star}	&		$1000$	&		$10$	&	$< 1$	& $< 1$	\\
			\midrule
			\myalign{l}{Lu}	& $[x_{1i};x_{2i};x_{3i};x_{1i}x_{3i};x_{1i}x_{2i};x_{2i}x_{3i}]$	&	\myalign{l}{Line}	&		$2000$	&		$150$	&		$< 1$ & $1.41$	\\
			\bottomrule
		\end{tabular}}
	\end{table*}
	
	\begin{figure*}[h!]
		\subfloat[\centering Evolution of open-loop network\label{fig:sc time ol}]{
			\includegraphics[width=0.33\linewidth]{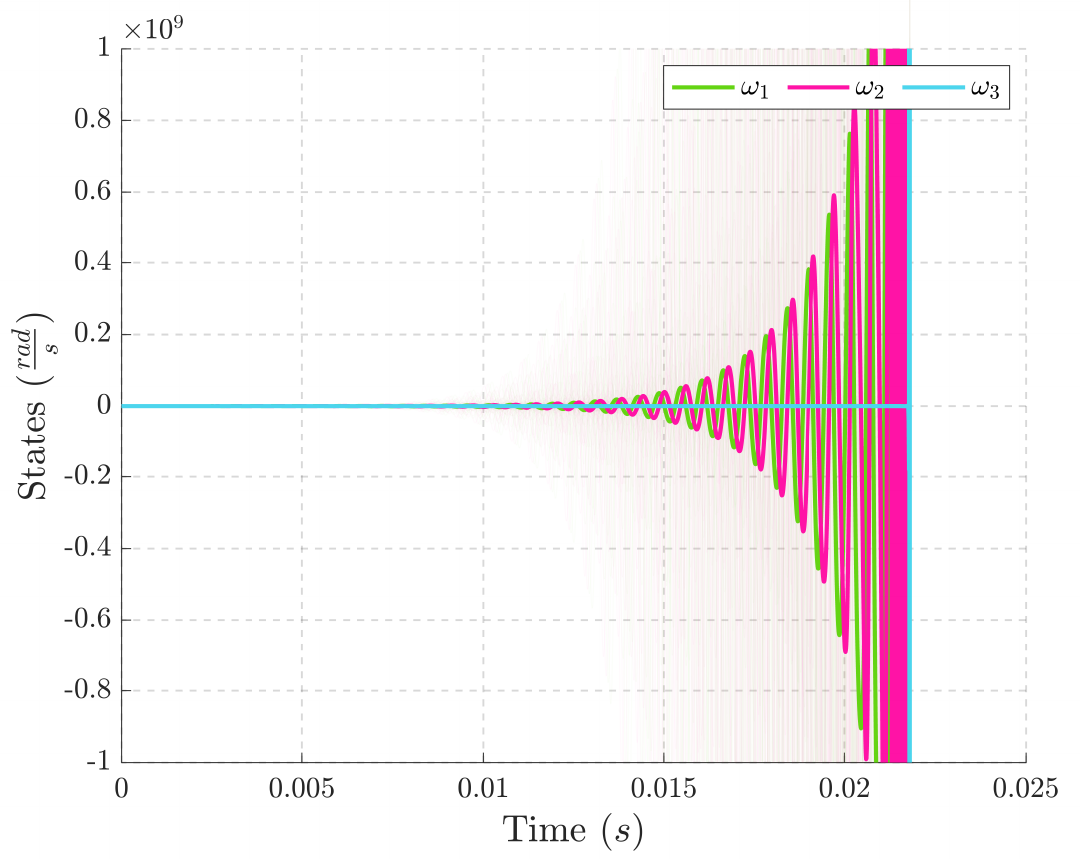}}
		\quad\quad\quad
		\subfloat[\centering 3D visualization of open-loop network\label{fig:sc 3D ol}]{
			\includegraphics[width=0.33\linewidth]{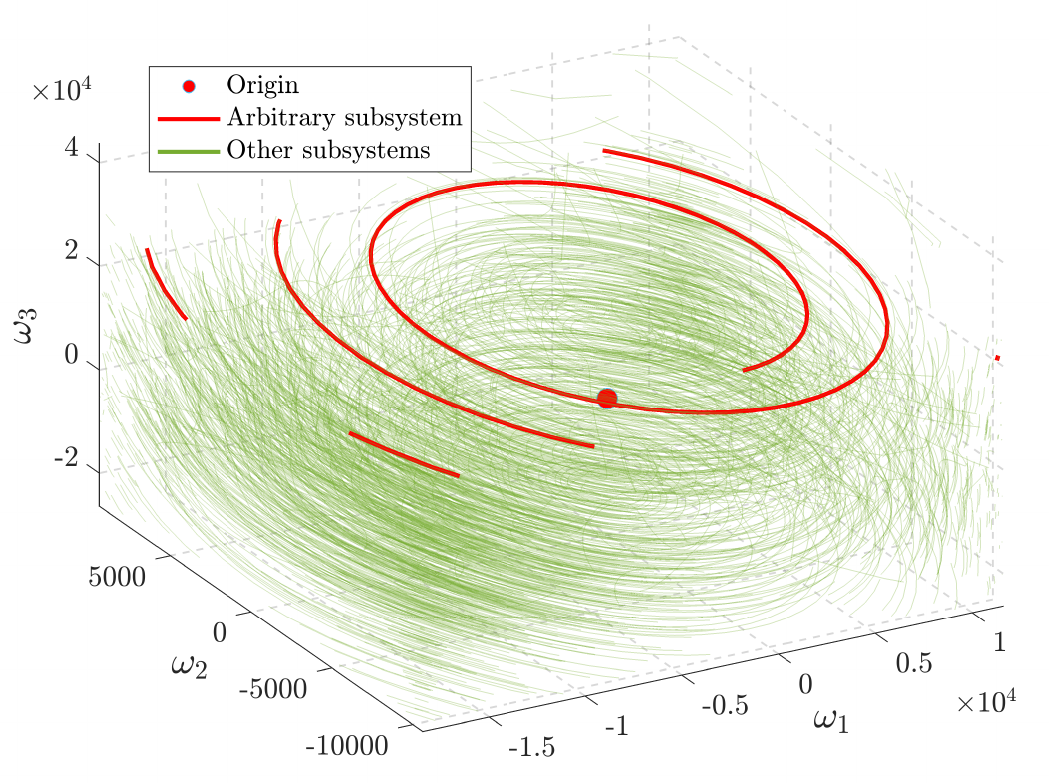}}
		\\
		\subfloat[\centering Evolution of closed-loop network \label{fig:sc fully time}]{
			\includegraphics[width=0.33\linewidth]{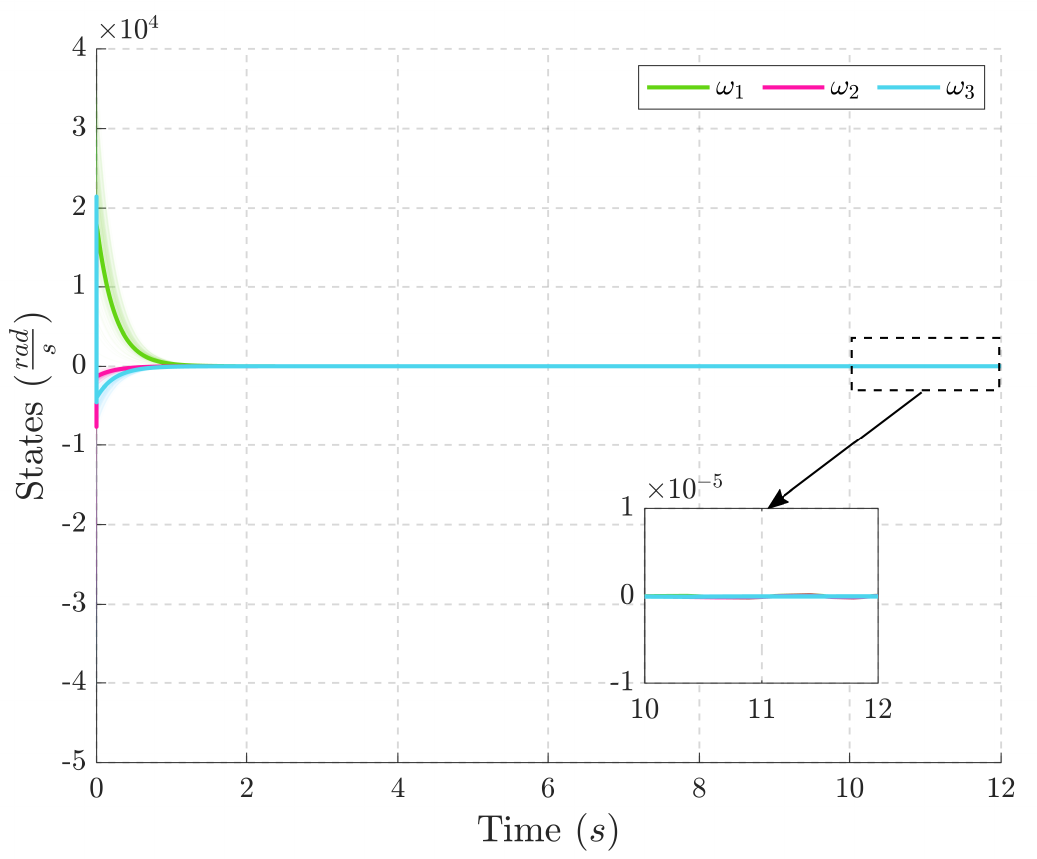}}
		\quad\quad\quad
		\subfloat[\centering 3D visualization of closed-loop network\label{fig:sc fully 3D}]{
			\includegraphics[width=0.33\linewidth]{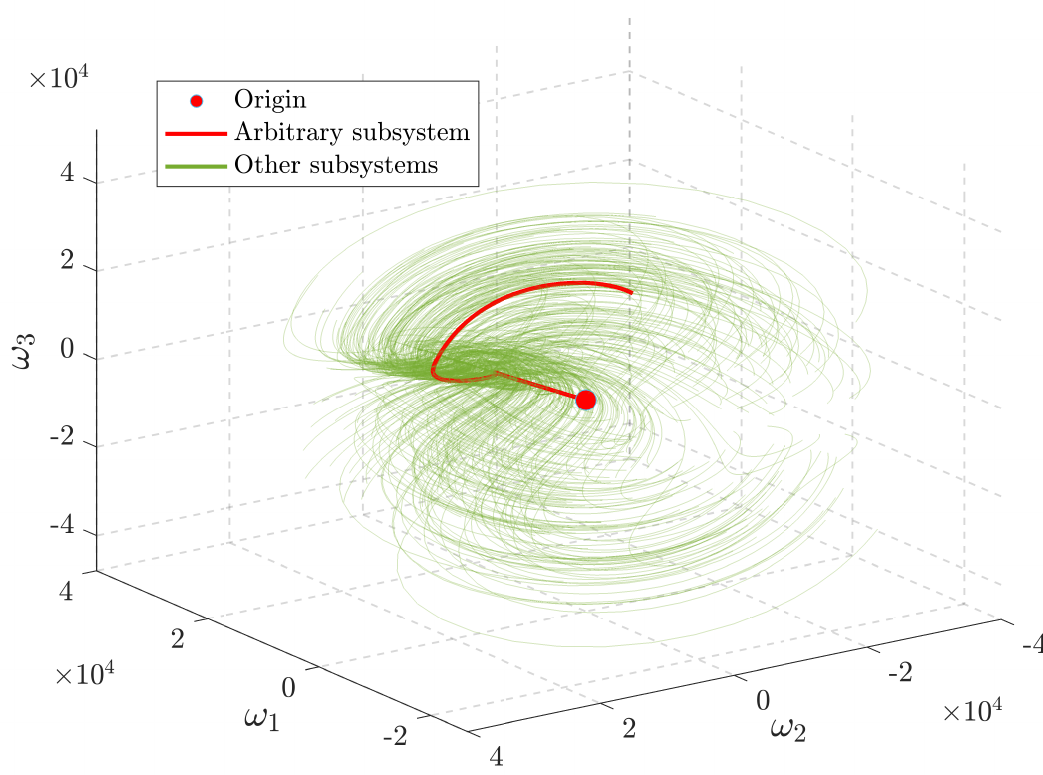}}
		\caption{Interconnected spacecraft network in a fully-interconnected topology: \protect\subref{fig:sc time ol} and \protect\subref{fig:sc 3D ol} depict the network's behavior without the controller, resulting in diverging trajectories. \protect\subref{fig:sc fully time} illustrates the network's evolution under the synthesized controller, highlighting the states of an arbitrary subsystem in bold and shading those of the remaining subsystems. \protect\subref{fig:sc fully 3D} visualizes trajectories of the network under the synthesized controller in 3D, demonstrating that all states converge to the origin.\label{fig:sc fully}}
	\end{figure*}
	
		\begin{figure*}[t!]
		\subfloat[\centering Evolution of closed-loop network\label{fig:sc ring time}]{
			\includegraphics[width=0.33\linewidth]{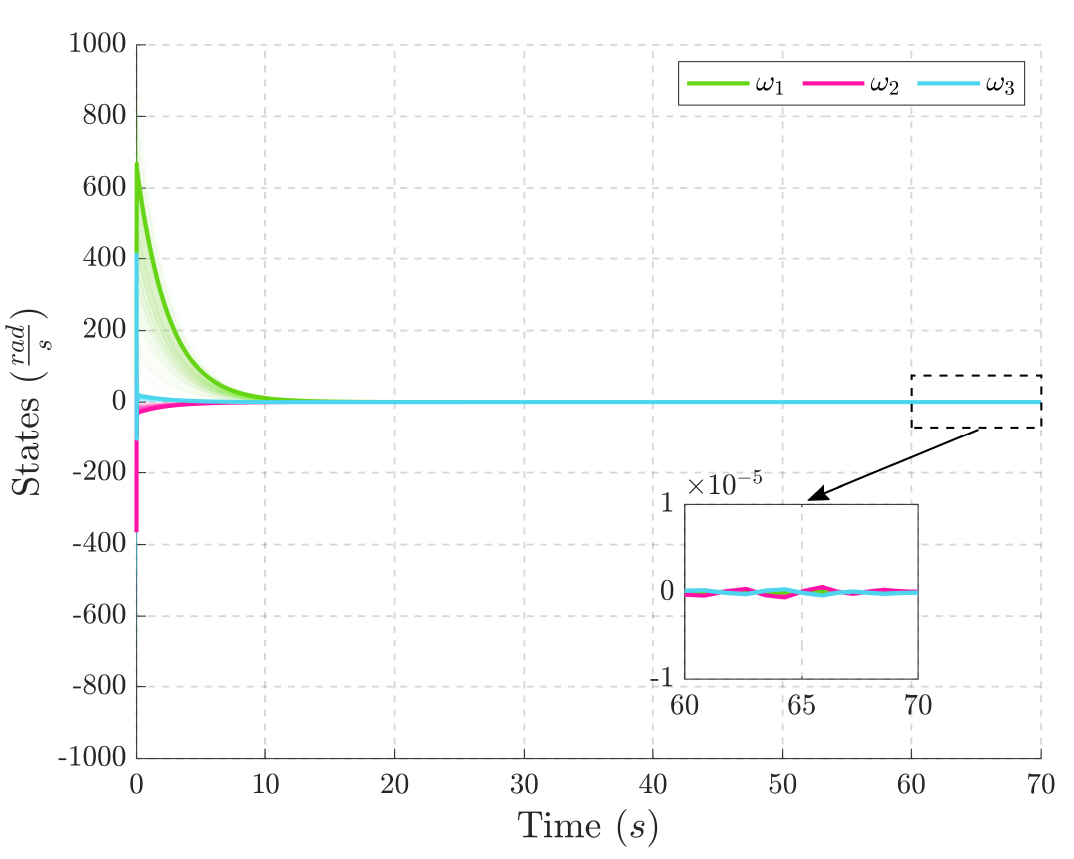}}
		\quad\quad\quad
		\subfloat[\centering 3D visualization of closed-loop network\label{fig:sc ring 3D}]{
			\includegraphics[width=0.33\linewidth]{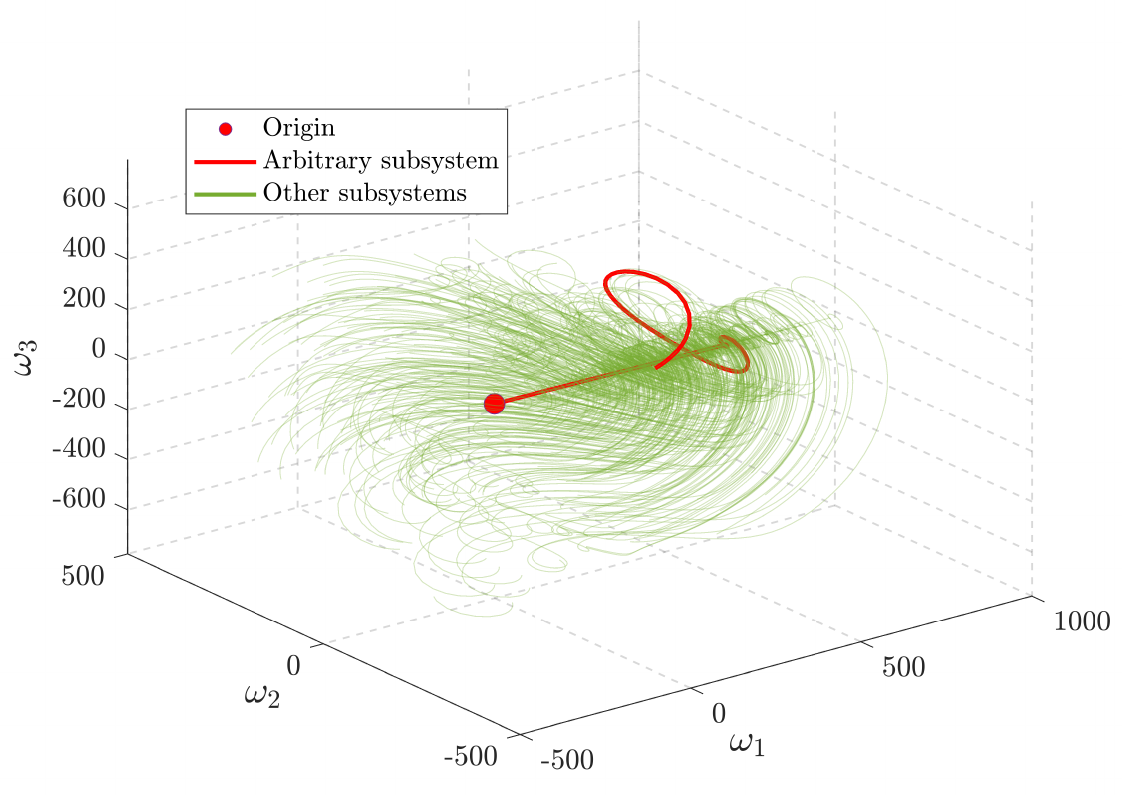}}
		\caption{Interconnected spacecraft network in a ring topology: \protect\subref{fig:sc ring time} demonstrates the evolution of the network under the synthesized controller, highlighting the states of an arbitrary subsystem in bold and shading those of the remaining subsystems. \protect\subref{fig:sc ring 3D} visualizes trajectories of the network under the synthesized controller in 3D, demonstrating that all states converge to the origin.\label{fig:sc ring}}
	\end{figure*}
	
	\subsection{Spacecraft Network: Fully-interconnected Topology}
	The first case study is a fully-interconnected network of $1000$  rotating rigid spacecraft\footnote{While we assume all subsystems are identical for simplicity, our proposed results are applicable to networks with heterogeneous subsystems.}~\citep{khalil2002control}. Each spacecraft has three states $x_i = [x_{1i};x_{2i};x_{3i}]$, which are the angular velocities $\omega_{1i}$ to $\omega_{3i}$ along the principal axes and are affected by those of its neighboring subsystems. Each subsystem is described by the following dynamics
	\begin{equation}\label{eq:sc}
		\begin{split}
			\dot x_{1i} &= \frac{(J_{2i} - J_{3i})}{J_{1i}}x_{2i}x_{3i} + \frac{1}{J_{1i}}u_{1i} -10^{-4}\sum_{\substack{j=1 \\ j\neq i}}^{1000}x_{3j},\\
			\dot x_{2i} &= \frac{(J_{3i} - J_{1i})}{J_{2i}}x_{1i}x_{3i} + \frac{1}{J_{2i}}u_{2i} -10^{-4}\sum_{\substack{j=1 \\ j\neq i}}^{1000}x_{1j},\\
			\dot x_{3i} &= \frac{(J_{1i} - J_{2i})}{J_{3i}}x_{1i}x_{2i} + \frac{1}{J_{3i}}u_{3i} -10^{-4}\sum_{\substack{j=1 \\ j\neq i}}^{1000}x_{2j},
		\end{split}
	\end{equation}
	where $u_i = [u_{1i};u_{2i};u_{3i}]$ is the torque input, and $J_{1i}$ to $J_{3i}$ are the principal moments of inertia. Considering $\mathcal F_i(x_i) = [x_{1i};x_{2i};x_{3i};$ $x_{2i}x_{3i};x_{1i}x_{3i};x_{1i}x_{2i}]$ as the actual monomials of each subsystem, unknown matrices
	\begin{align*}
		A_i &= \begin{bmatrix}
			0 & 0 & 0 &  \frac{(J_{2i} - J_{3i})}{J_{1i}} & 0 & 0 \\
			0 & 0 & 0 & 0 & \frac{(J_{3i} - J_{1i})}{J_{2i}} & 0\\
			0 & 0 & 0 & 0 & 0 & \frac{(J_{1i} - J_{2i})}{J_{3i}}
		\end{bmatrix},\\
		B_i &= \begin{bmatrix}
			\frac{1}{J_{1i}} & 0 & 0\\
			0 & \frac{1}{J_{2i}}& 0\\
			0 & 0 & \frac{1}{J_{3i}}
		\end{bmatrix},
	\end{align*}
	and the adversarial matrix partitions
	\begin{align*}
		D_{ij} &= 10^{-4}\times\begin{bmatrix}
			0 & 0 & -1\\
			-1 & 0 & 0\\
			0 & -1 & 0
		\end{bmatrix},
	\end{align*}
	according to Definition~\ref{def:interconnection} for a fully-interconnected topology, we have
	\begin{align*}
		A &= \begin{bmatrix}
			A_1 & \tilde D_{12} & \tilde D_{13} & \cdots & \tilde D_{1\,1000}\\
			\tilde D_{21} & A_2 &  \tilde D_{23} & \cdots & \tilde D_{2\,1000}\\
			{} & {} & {} & {} & {}\\
			\vdots & {} & \ddots & {} & \vdots\\
			{} & {} & {} & {} & {}\\
			\tilde D_{999\,1} & \cdots &  \tilde D_{999\,998} & A_{999} & \tilde D_{999\,1000}\\
			\tilde D_{1000\,1} & \cdots & \tilde D_{1000\,998} & \tilde D_{1000\,999} & A_{1000}
		\end{bmatrix}\!,\\
		B &= \mathsf{blkdiag}(B_1,\ldots,B_{1000}).
	\end{align*}
	
	Figures~\ref{fig:sc fully}\protect\subref{fig:sc time ol} and \ref{fig:sc fully}\protect\subref{fig:sc 3D ol} depict the network's behavior in the absence of a controller (open-loop). To avoid such behavior, we construct a CLF and its associated controller utilizing our data-driven framework. To do so, we gather $T = 100$ samples with $\tau=0.01$ from each subsystem~\eqref{eq:sc} and choose $\mathcal F_i(x_i) = [x_{1i};x_{2i};x_{3i};x_{1i}^2;x_{2i}^2;x_{3i}^2;x_{1i}x_{2i};\\
	x_{1i}x_{3i};$ $x_{2i}x_{3i}]$. Following the steps in Algorithm~\ref{Alg:1}, we set $\kappa_i = 0.1,\, \pi_i = 0.71578$ and design

	\begin{align*}
		P_i =& \begin{bmatrix}
			4.8315  &  0.9437 &  -0.0149\\
			0.9437  &  3.4773  & -0.2598\\
			-0.0149 &  -0.2598 &   2.8181
		\end{bmatrix}\!,\\
		\mathcal V_i(x_i) =& ~4.8315x_{1i}^2+1.8873x_{1i}x_{2i}-0.029784x_{1i}x_{3i}+3.4773x_{2i}^2\\
		&-0.51959x_{2i}x_{3i}+2.8181x_{3i}^2,\\
		u_{1i} =&-7.2915x_{1i}^2 + 0.095205x_{1i}x_{2i} - 30.3111x_{1i}x_{3i} + 99.6962x_{2i}^2\\
		 &- 33.8855x_{2i}x_{3i} + 46.4852x_{3i}^2 - 812.734x_{1i} - 124.5246x_{2i}\\
		 &+ 6.5823x_{3i},\\
		u_{2i} = &~37.0636x_{1i}^2 - 135.3413x_{1i}x_{2i} + 115.2349x_{1i}x_{3i} - 27.1295x_{2i}^2\\
		& + 34.5258x_{2i}x_{3i} + 0.39821x_{3i}^2 - 206.8314x_{1i} - 605.312x_{2i}\\
		& + 46.1622x_{3i},\\
		u_{3i} =& -25.567x_{1i}^2 - 83.3144x_{1i}x_{2i} - 90.2213x_{1i}x_{3i} - 1.4901x_{2i}^2\\
		 &- 20.3337x_{2i}x_{3i} + 0.42353x_{3i}^2 - 29.969x_{1i} + 57.9927x_{2i}\\
		 & - 710.6459x_{3i},
	\end{align*}
	with $\underline{\alpha}_i = 2.6603, \overline{\alpha}_i = 5.3229, \rho_i = 1.3957\times10^{-5}$, which satisfies the compositional condition~\eqref{eq:comp con}. Hence, the CLF and GAS controller for the network are obtained as
	\begin{align*}
		\mathcal V(x) &= \sum_{i = 1}^{1000}\mathcal V_i(x_i), ~ u = [u_1;\dots;u_{1000}], 
	\end{align*}
	with $\underline{\alpha} = \underset{i}{\min}\{\underline{\alpha}_i\} = 2.6603, \overline{\alpha} = \underset{i}{\max}\{\overline{\alpha}_i\} = 5.3229,$ and $\kappa = 0.0948$. We apply the GAS controller to the network with arbitrary initial conditions $x(0)\in[-25000,25000]$, which are large enough to evidence the efficacy of our proposed approach. Figures~\ref{fig:sc fully}\protect\subref{fig:sc fully time} and \ref{fig:sc fully}\protect\subref{fig:sc fully 3D} show the closed-loop trajectories of a representative spacecraft subsystem in a fully-interconnected topology under the synthesized controller.
	\begin{remark}
		In our case studies, we assumed error-free state derivative data. However, if state derivatives are approximated, the imposed error ($\Delta_i\Delta_i^\top$)  is bounded by ($\Lambda_i\Lambda_i^\top$) (as discussed in Subsection~\ref{sub:lim}). In this case, \eqref{eq:con2 SOS} should be adjusted to incorporate the error, ensuring robustness at the cost of increased conservatism in satisfying the condition.
	\end{remark}
	\subsection{Spacecraft Network: Ring Topology}
	Here, we consider a ring network of $2000$ spacecrafts. Each subsystem is described by
		\begin{equation}\label{eq:sc ring}
			\begin{split}
				\dot x_{1i} &= \frac{(J_{2i} - J_{3i})}{J_{1i}}x_{2i}x_{3i} + \frac{1}{J_{1i}}u_{1i} -10^{-4}x_{3(i - 1)},\\
				\dot x_{2i} &= \frac{(J_{3i} - J_{1i})}{J_{2i}}x_{1i}x_{3i} + \frac{1}{J_{2i}}u_{2i} -10^{-4}x_{1(i-1)},\\
				\dot x_{3i} &= \frac{(J_{1i} - J_{2i})}{J_{3i}}x_{1i}x_{2i} + \frac{1}{J_{3i}}u_{3i} -10^{-4}x_{2(i - 1)},
			\end{split}
		\end{equation}
		for $i\in\{2,\dots,2000\}$, where the first subsystem is affected by the last subsystem.
		The matrices of the ring network as per Definition~\ref{def:interconnection} are as follows:
		\begin{align*}
			A &= \begin{bmatrix}
				A_1 & \mathbf{0}_{3\times6} & \mathbf{0}_{3\times6} & \cdots & \tilde D_{1\,2000}\\
				\tilde D_{21} & A_2 & \mathbf{0}_{3\times6} & \cdots & \mathbf{0}_{3\times6}\\
				{} & {} & {} & {} & {}\\
				\vdots & {} & \ddots & {} & \vdots\\
				{} & {} & {} & {} & {}\\
				\mathbf{0}_{3\times6} & \cdots &  \tilde D_{1999\,1998} & A_{1999} & \mathbf{0}_{3\times6}\\
				\mathbf{0}_{3\times6} & \cdots & \mathbf{0}_{3\times6} & \tilde D_{2000\,1999} & A_{2000}
			\end{bmatrix},\\
			B &= \mathsf{blkdiag}(B_1,\ldots,B_{2000}).
		\end{align*}
	 We collect $T = 15$ samples with a sampling time of $\tau = 0.01$ from each subsystem in~\eqref{eq:sc ring}. Following the procedure outlined in Algorithm~\ref{Alg:1}, we set $\kappa_i = 10^{-6}$ and $\pi_i = 0.00077986$, obtaining the following results:
	\begin{align*}
		P_i =& \begin{bmatrix}
			318.1693 & -19.1415 &  49.8083\\
			-19.1415 & 857.7995 & 184.3652\\
			49.8083 & 184.3652 & 505.8565
		\end{bmatrix},\\
		\mathcal V_i(x_i) =& ~318.1693x_{1i}^2 - 38.2831x_{1i}x_{2i} + 99.6167x_{1i}x_{3i} \\
		&+ 857.7995x_{2i}^2 + 368.7303x_{2i}x_{3i} + 505.8565x_{3i}^2,\\
		u_{1i} =& ~0.11777x_{1i}^2 + 16.3415x_{1i}x_{2i} + 12.7878x_{1i}x_{3i} + 80.6853x_{2i}^2 \\
		&+ 63.0886x_{2i}x_{3i} + 15.6556x_{3i}^2 - 77.3744x_{1i} + 5.4254x_{2i} \\
		&+ 9.8153x_{3i},\\
		u_{2i} = &-2.6059x_{1i}^2 - 16.1453x_{1i}x_{2i} - 40.3623x_{1i}x_{3i} + 4.0518x_{2i}^2\\
		& + 4.9017x_{2i}x_{3i} + 17.1156x_{3i}^2 + 12.9201x_{1i} - 206.6947x_{2i} \\
		&+ 19.7909x_{3i},\\
		u_{3i} = &-2.6307x_{1i}^2 - 88.7645x_{1i}x_{2i} - 47.1421x_{1i}x_{3i} - 15.7124x_{2i}^2 \\
		&- 35.6218x_{2i}x_{3i} - 11.6692x_{3i}^2 - 3.8432x_{1i} + 107.8112x_{2i} \\
		&- 141.9266x_{3i},
	\end{align*}
	with $\underline{\alpha}_i = 296.3880, \overline{\alpha}_i = 936.7001, \rho_i = 1.2823\times10^{-5}$. These parameters satisfy the compositional condition~\eqref{eq:comp con}, yielding $\underline{\alpha} = \underset{i}{\min}\{\underline{\alpha}_i\} = 296.3880$, $\overline{\alpha} = \underset{i}{\max}\{\overline{\alpha}_i\} = 936.7001$, and $\kappa = 9.5674 \times 10^{-7}$. Consequently, the CLF and GAS controller for the network are given by
	\begin{align*}
		\mathcal V(x) &= \sum_{i = 1}^{2000}\mathcal V_i(x_i), ~ u = [u_1;\dots;u_{2000}].
	\end{align*}
	We apply the GAS controller to the network with arbitrary initial conditions $x(0)\in[-400,400]$. Figures~\ref{fig:sc ring}\protect\subref{fig:sc ring time} and \ref{fig:sc ring}\protect\subref{fig:sc ring 3D} show the closed-loop trajectories of a representative spacecraft subsystem in a ring topology under the synthesized controller, verifying the efficacy of our approach.
	\begin{figure*}[t!]
		\subfloat[\centering Evolution of closed-loop network\label{fig:academic time}]{
			\includegraphics[width=0.33\linewidth]{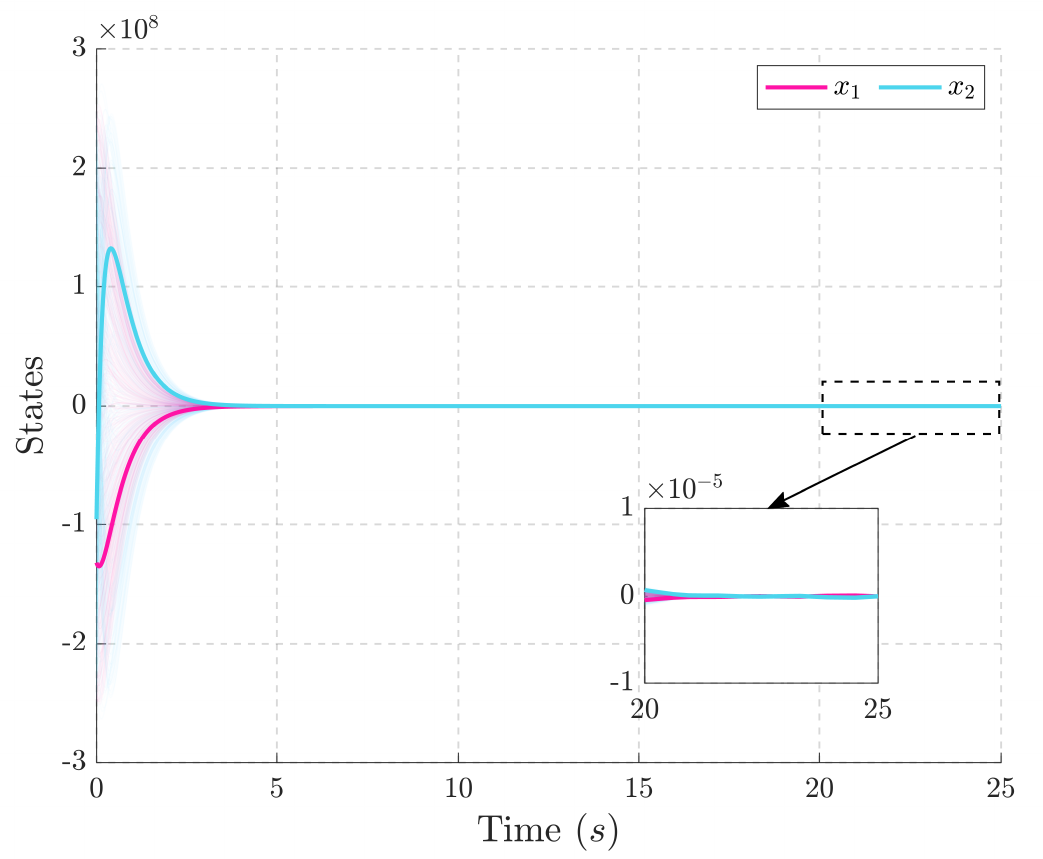}}
	\quad\quad\quad
		\subfloat[\centering ISS Lyapunov function $\mathcal V_i(x_i)$\label{fig:CLF}]{
			\includegraphics[width=0.33\linewidth]{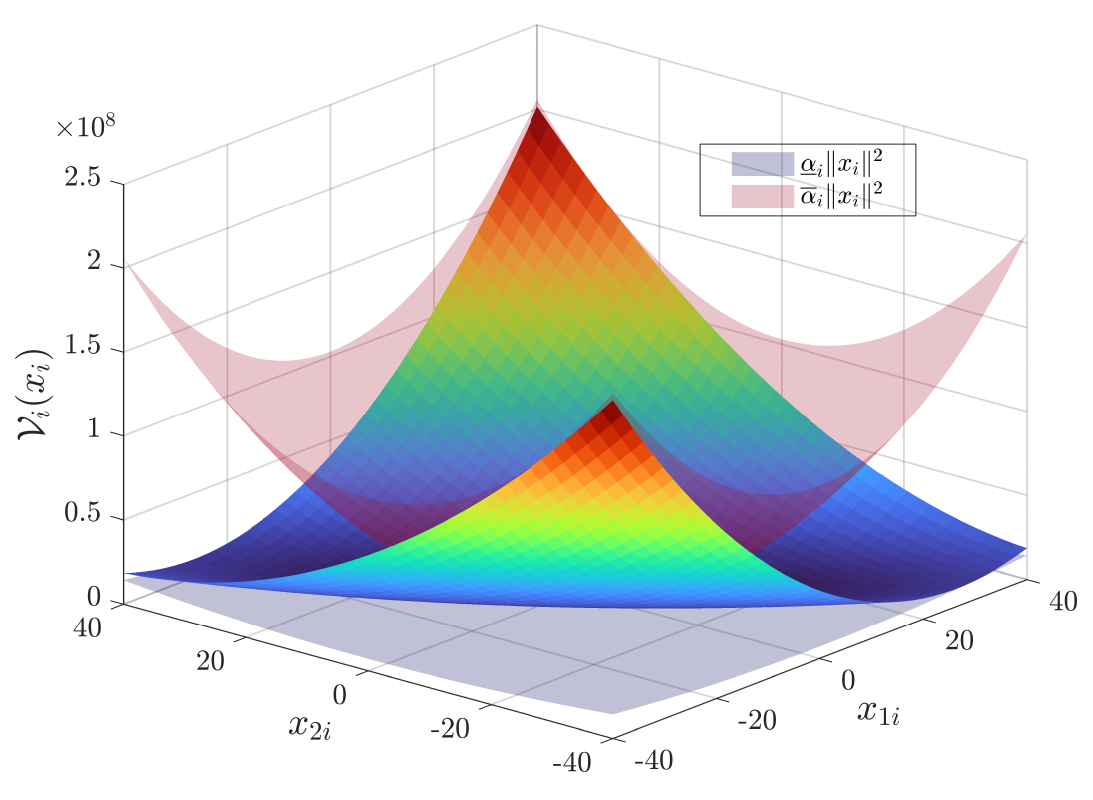}}
		\caption{Interconnected academic network in a binary tree topology: \protect\subref{fig:academic time} demonstrates the evolution of the network under the synthesized controller, highlighting the states of an arbitrary subsystem in bold and shading those of the remaining subsystems; \protect\subref{fig:CLF} illustrates the ISS Lyapunov function $\mathcal V_i(x_i)$ of an arbitrary subsystem, demonstrating the fulfillment of condition~\eqref{eq:ISS-con1}.\label{fig:academic}}
	\end{figure*}
	
	\subsection{Academic Network: Binary Tree Topology}
Another case study involves a binary tree network of $2047$ subsystems of an academic example. Each subsystem has two states $x_i = [x_{1i};x_{2i}]$, and is described by the following dynamics:
\begin{equation}\label{eq:academic}
	\begin{split}
		\dot x_{1i} &= x_{2i} + 10^{-2}\Pi(i = 2j \vee i = 2j + 1) x_{1j},\\
		\dot x_{2i} & = x_{1i}^2 + u_i,
	\end{split}
\end{equation}
where $u_i$ is the control input, and $\Pi$ is an indicator function defined as
\begin{align*}
	\Pi(i = 2j \vee i = 2j + 1) = \begin{cases}
		1 & i = 2j \vee i = 2j + 1,\\
		0 & \text{Otherwise}.
	\end{cases}
\end{align*}
with $\vee$ being the logical \texttt{OR} operation. Then, considering $\mathcal F_i(x_i) = [x_{1i};x_{2i};x_{1i}^2]$ as the actual monomials of the subsystems, unknown matrices
\begin{align*}
	A_i = \begin{bmatrix}
		0 & 1 & 0\\
		0 & 0 & 1
	\end{bmatrix}, \quad B_i = \begin{bmatrix}
		0\\
		1
	\end{bmatrix},
\end{align*}
and the adversarial matrix partitions
\begin{align*}
	D_{ij} = 10^{-2}\times\begin{bmatrix}
		1 & 0\\
		0 & 0
	\end{bmatrix}\!,
\end{align*}
the network matrices based on Definition~\ref{def:interconnection} are as follows 
\begin{align*}
	A & = \{A_{ij}\} = \begin{cases}
		A_i & i = j,\\
		\tilde D_{ij} & i = 2j \vee i = 2j + 1,\\
		\mathbf{0}_{2\times 3} & \text{Otherwise,}
	\end{cases}\\
	B &= \mathsf{blkdiag}(B_1,\ldots,B_{2047}).
\end{align*}
We gather $T = 12$ samples with $\tau=0.01$ from each subsystem~\eqref{eq:academic} and select $\mathcal F_i(x_i) = [x_{1i};x_{2i};x_{1i}^2;x_{2i}^2;x_{1i}x_{2i}]$. Following the procedure outlined in Algorithm~\ref{Alg:1}, we set $\kappa_i = 0.1,\, \pi_i = 0.00014535$, obtaining

\begin{align*}
	P_i =& ~10^{4}\times\begin{bmatrix}
		4.3254  &  2.8599\\
		2.8599  &  2.5547
	\end{bmatrix},\\
	\mathcal V_i(x_i) =& ~43253.98x_{1i}^2 + 57198.4261x_{1i}x_{2i} + 25546.7906x_{2i}^2,\\
	u_{i} =&-x_{1i}^2 - 1.288\times 10^{-10}x_{1i}x_{2i} - 2.5113\times 10^{-11}x_{2i}^2\\
	& - 8.4045x_{1i} - 6.6181x_{2i},
\end{align*}
with $\underline{\alpha}_i = 4.4621\times10^3, \overline{\alpha}_i = 6.4339\times10^4, \rho_i = 0.68799$, satisfying the compositional condition~\eqref{eq:comp con}. Thus, the CLF and GAS controller for the network are designed as:
\begin{align*}
	\mathcal V(x) &= \sum_{i = 1}^{2047}\mathcal V_i(x_i), ~ u = [u_1;\dots;u_{2047}],
\end{align*}
with $\underline{\alpha} = \underset{i}{\min}\{\underline{\alpha}_i\} = 4.4621\times10^3, \overline{\alpha} = \underset{i}{\max}\{\overline{\alpha}_i\} = 6.4339\times10^4, \kappa = 0.0997$. We apply the GAS controller to the network, starting from arbitrary initial conditions $x(0)\in[-2.5\times 10^{8},2.5\times 10^{8}]$, which are sufficiently large to demonstrate the effectiveness of the proposed approach.
Figure~\ref{fig:academic}\protect\subref{fig:academic time} depicts the closed-loop trajectories of a representative subsystem under the synthesized controller, while Figure~\ref{fig:academic}\protect\subref{fig:CLF} visualizes the satisfaction of condition~\eqref{eq:ISS-con1}.

\subsection{Academic Network: Star Topology}
In this subsection, we consider a network of academic dynamics, with a star topology, comprising $1000$ subsystems. Each subsystem has the following dynamics
\begin{subequations}\label{eq:academic star}
	\begin{equation}
		\begin{split}
			\dot x_{1i} &= x_{2i} + 10^{-2} x_{11},\\
			\dot x_{2i} & = x_{1i}^2 + u_i,
		\end{split}
	\end{equation}
	for $i\in\{2,\dots,1000\}$, whereas for $i = 1$
	\begin{equation}
		\begin{split}
			\dot x_{11} &= x_{21},\\
			\dot x_{21} & = x_{11}^2 + u_1.
		\end{split}
	\end{equation}
\end{subequations}
Hence, in compliance with Definition~\ref{def:interconnection}, the network matrices are
\begin{align*}
	A = \{A_{ij}\} = \begin{cases}
		A_i & i = j,\\
		\tilde D_{ij} & j = 1,\,i\neq j,\\
		\mathbf{0}_{2\times 3} & \text{Otherwise},
	\end{cases}\quad B = \mathsf{blkdiag}(B_1,\ldots,B_{1000}).
\end{align*}
We gather $T = 10$ samples with $\tau=0.01$ from each subsystem~\eqref{eq:academic star}. Following the steps in Algorithm~\ref{Alg:1}, we set $\kappa_i = 1,\, \pi_i = 0.0001661$ and obtain

\begin{align*}
	P_i =& 10^{4}\times\begin{bmatrix}
		4.7642  &  1.0768\\
		1.0768  &  0.3979
	\end{bmatrix},\\
	\mathcal V_i(x_i) =& ~47641.9535x_{1i}^2 + 21535.1163x_{1i}x_{2i} + 3979.1603x_{2i}^2,\\
	u_{i} =&-x_{1i}^2 + 5.9826\times 10^{-11}x_{1i}x_{2i} - 2.5598\times 10^{-11}x_{2i}^2 \\
	&- 41.2568x_{1i} - 11.7528x_{2i},
\end{align*}
with $\underline{\alpha}_i = 1.4682\times10^3, \overline{\alpha}_i = 5.0153\times10^4, \rho_i = 0.60205$, which satisfies the compositional condition~\eqref{eq:comp con}. Hence, the CLF and GAS controller for the network are designed as
\begin{align*}
	\mathcal V(x) &= \sum_{i = 1}^{1000}\mathcal V_i(x_i), ~ u = [u_1;\dots;u_{1000}],
\end{align*}
with $\underline{\alpha} = \underset{i}{\min}\{\underline{\alpha}_i\} = 1.4682\times10^3, \overline{\alpha} = \underset{i}{\max}\{\overline{\alpha}_i\} = 5.0153\times10^4, \kappa = 0.5903$. We apply the GAS controller to the network, initialized with arbitrary initial conditions $x(0)\in[-2.5\times 10^8,2.5\times 10^8]$, which are sufficiently large to demonstrate our proposed approach.
Figure~\ref{fig:academic star}\protect\subref{fig:academic star time} depicts the closed-loop trajectories of a representative subsystem under the synthesized controller, while Figure~\ref{fig:academic star}\protect\subref{fig:CLF star} visualizes the satisfaction of condition~\eqref{eq:ISS-con1}.

	\begin{figure*}[t!]
	\subfloat[\centering Evolution of closed-loop network\label{fig:academic star time}]{
		\includegraphics[width=0.33\linewidth]{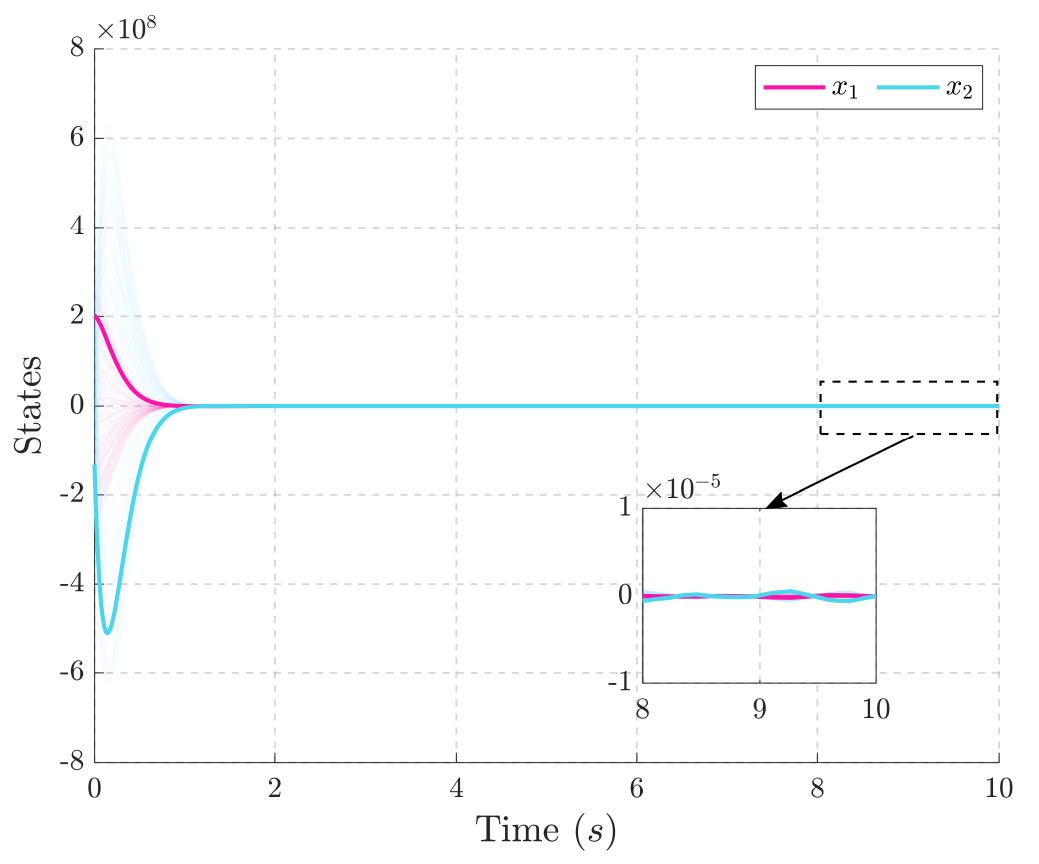}}
	\quad\quad\quad
	\subfloat[\centering ISS Lyapunov function $\mathcal V_i(x_i)$\label{fig:CLF star}]{
		\includegraphics[width=0.33\linewidth]{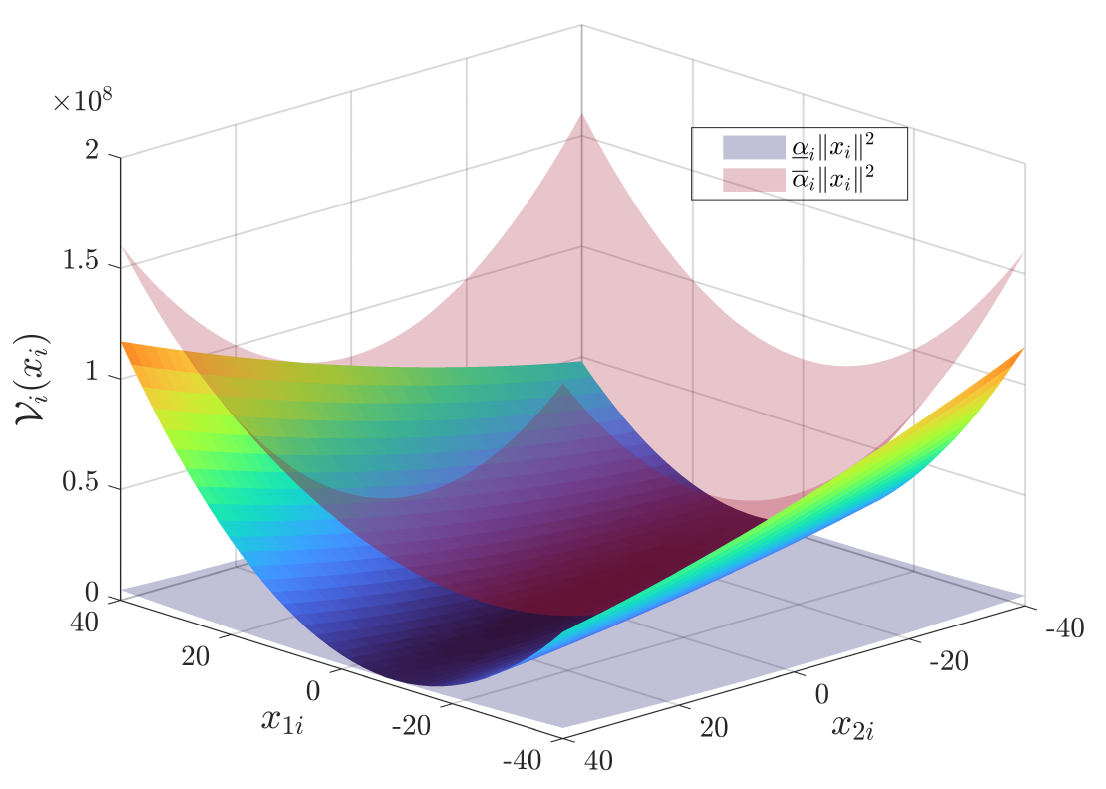}}
	\caption{Interconnected academic network in a star topology: \protect\subref{fig:academic star time} demonstrates the evolution of the network under the synthesized controller, highlighting the states of an arbitrary subsystem in bold and shading those of the remaining subsystems; \protect\subref{fig:CLF star} illustrates the ISS Lyapunov function $\mathcal V_i(x_i)$ of an arbitrary subsystem, demonstrating the fulfillment of condition~\eqref{eq:ISS-con1}.\label{fig:academic star}}
\end{figure*}

\begin{figure*}[t!]
	\subfloat[\centering Evolution of closed-loop network \label{fig:lu time}]{
		\includegraphics[width=0.33\linewidth]{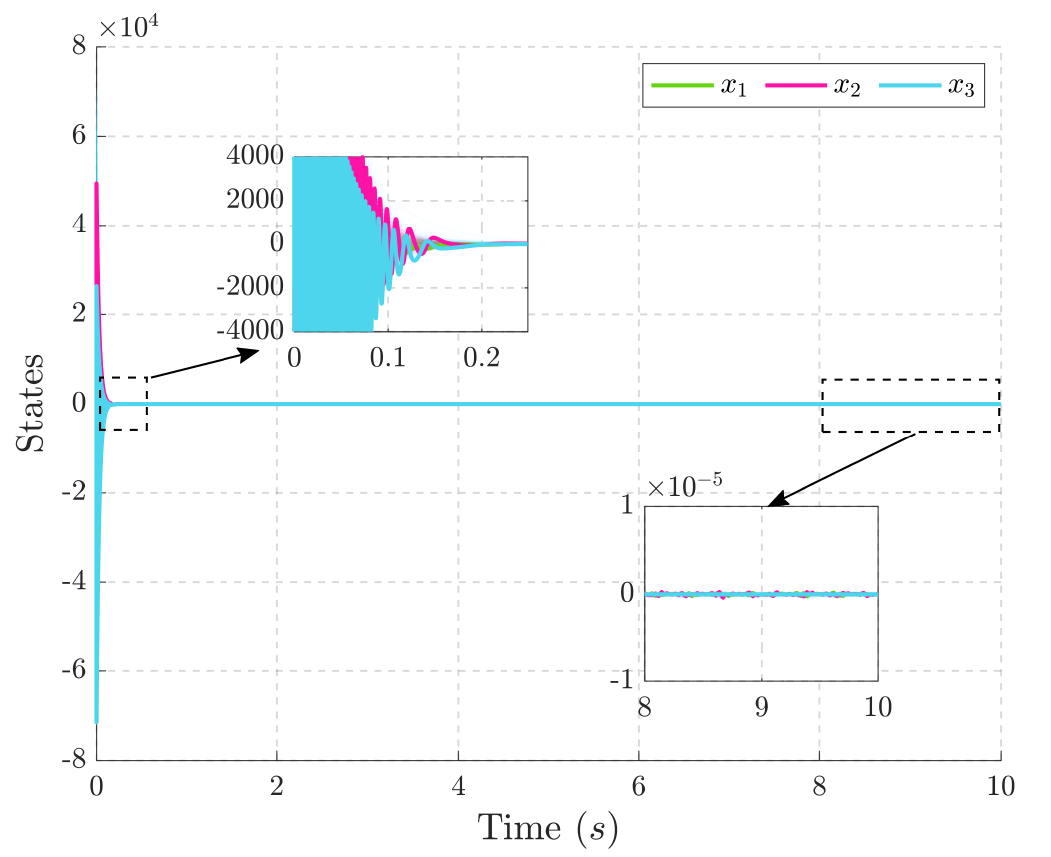}}
	\quad\quad\quad
	\subfloat[\centering 3D visualization of closed-loop network\label{fig:lu 3D}]{
		\includegraphics[width=0.33\linewidth]{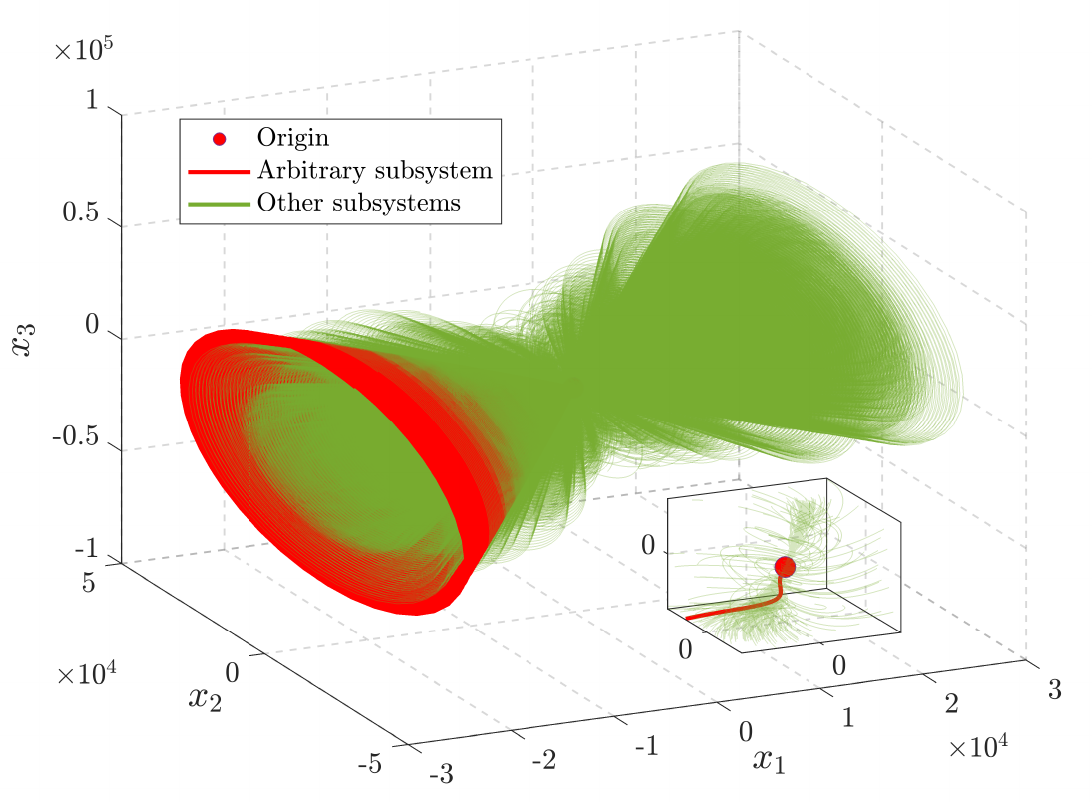}}
	\caption{Interconnected Lu network in a line topology: \protect\subref{fig:lu time} demonstrates the evolution of the network under the synthesized controller, highlighting the states of an arbitrary subsystem in bold and shading those of the remaining subsystems. \protect\subref{fig:lu 3D} visualizes trajectories of the network under the synthesized controller in 3D, demonstrating that all states converge to the origin.\label{fig:lu}}
\end{figure*}

\subsection{Lu Network: Line Topology}
As for the last case study, we consider a line network of $2000$ Lorenz-type systems~\citep{strogatz2018nonlinear}, which are well-known for their chaotic solutions. Lorenz networks are used in various fields, such as climate modeling, where they help understand the dynamics of coupled atmospheric processes. They are also employed in engineering to design robust systems that can tolerate and leverage chaotic behavior, in secure communications to develop encryption schemes, and in neuroscience to model complex brain activities.
The Lorenz-type systems considered here are known as Lu. Each subsystem has three states $x_i = [x_{1i};x_{2i};x_{3i}]$, evolving as follows
\begin{subequations}\label{new76}
	\begin{equation}
		\begin{split}
			\dot x_{1i} &= - 36x_{1i} + 36x_{2i} - 10^{-3} x_{1(i - 1)},\\
			\dot x_{2i} &= 28x_{2i} - x_{1i}x_{3i} + u_i - 10^{-3} x_{2(i - 1)},\\
			\dot x_{3i} &= - 20x_{3i} + x_{1i}x_{2i} - 10^{-3} x_{3(i - 1)},
		\end{split}
	\end{equation}
	for $i\in\{2,\dots,2000\}$ while subsystem $i = 1$ evolves as
	\begin{equation}
		\begin{split}
			\dot x_{11} &= - 36x_{11} + 36x_{21},\\
			\dot x_{21} &= 28x_{21} - x_{11}x_{31} + u_1,\\
			\dot x_{31} &= - 20x_{31} + x_{11}x_{21}.
		\end{split}
	\end{equation}
\end{subequations}
Therefore, with actual monomials $\mathcal F_i(x_i) = [x_{1i};x_{2i};x_{3i};x_{1i}x_{3i};$ $x_{1i}x_{2i}]$, unknown system matrix and control matrix are
\begin{align*}
	A_i =\begin{bmatrix}
		-36  &  36  &   0  &   0  &   0\\
		0  &  28  &   0  &  -1  &   0\\
		0  &   0  & -20  &   0  &   1
	\end{bmatrix},\quad B_i =\begin{bmatrix}
		0\\
		1\\
		0
	\end{bmatrix},
\end{align*}
and adversarial matrix partitions are
\begin{align*}
	D_{ij} = 10^{-3} \times \begin{bmatrix}
		-1 & 0 & 0\\
		0 & -1 & 0\\
		0 & 0 & -1
	\end{bmatrix}\!.
\end{align*}
The matrices of the line network as per Definition~\ref{def:interconnection} are as follows:
\begin{align*}
	A &= \begin{bmatrix}
		A_1 & \mathbf{0}_{3\times5} & \mathbf{0}_{3\times5} & \cdots & \mathbf{0}_{3\times5}\\
		\tilde D_{21} & A_2 & \mathbf{0}_{3\times5} & \cdots & \mathbf{0}_{3\times5}\\
		{} & {} & {} & {} & {}\\
		\vdots & {} & \ddots & {} & \vdots\\
		{} & {} & {} & {} & {}\\
		\mathbf{0}_{3\times5} & \cdots &  \tilde D_{1999\,1998} & A_{1999} & \mathbf{0}_{3\times5}\\
		\mathbf{0}_{3\times5} & \cdots & \mathbf{0}_{3\times5} & \tilde D_{2000\,1999} & A_{2000}
	\end{bmatrix},\\
	B &= \mathsf{blkdiag}(B_1,\ldots,B_{2000}).
\end{align*}
We collect $T = 150$ samples with a sampling time of $\tau = 10^{-4}$ from each subsystem in~\eqref{new76} and set $\mathcal F_i(x_i) = [x_{1i};x_{2i};x_{3i};x_{1i}x_{3i};x_{1i}x_{2i};$ $x_{2i}x_{3i}]$. Following the procedure outlined in Algorithm~\ref{Alg:1}, we set $\kappa_i = 10^{-4}$ and $\pi_i = 0.1304$, obtaining the following results:
\begin{align*}
	P_i =& \begin{bmatrix}
		154.5767 & -3.0842\times 10^{-8} &  -82.5652\\
		-3.0842\times 10^{-8} & 86.8543 & -5.0350\times 10^{-8}\\
		-82.5652 & -5.0350\times 10^{-8} & 88.5103
	\end{bmatrix}\!,\\
	\mathcal V_i(x_i) =& ~154.5767x_{1i}^2 - 6.1684\times 10^{-8}x_{1i}x_{2i} - 165.1304x_{1i}x_{3i}\\
	& + 86.8543x_{2i}^2 - 1.007\times 10^{-7}x_{2i}x_{3i} + 88.5103x_{3i}^2,\\
	u_{i} =& ~0.95062x_{1i}^2 - 1.2347\times 10^{-8}x_{1i}x_{2i} - 0.019067x_{1i}x_{3i}\\
	& - 1.2869\times 10^{-7}x_{2i}^2 - 1.3213\times 10^{-7}x_{2i}x_{3i} \\
	&+ 2.1723\times 10^{-9}x_{3i}^2 - 64.0693x_{1i} - 77.1287x_{2i} \\
	&+ 34.2218x_{3i},
\end{align*}
with $\underline{\alpha}_i = 32.6154, \overline{\alpha}_i = 210.4716, \rho_i = 7.671137\times10^{-6}$. These parameters satisfy the compositional condition~\eqref{eq:comp con}, yielding $\underline{\alpha} = \underset{i}{\min}\{\underline{\alpha}_i\} = 32.6154$, $\overline{\alpha} = \underset{i}{\max}\{\overline{\alpha}_i\} = 210.4716$, and $\kappa = 9.9765 \times 10^{-5}$. Consequently, the CLF and GAS controller for the network are designed as
\begin{align*}
	\mathcal V(x) &= \sum_{i = 1}^{2000}\mathcal V_i(x_i), ~ u = [u_1;\dots;u_{2000}].
\end{align*}
We apply the GAS controller to the network with arbitrary initial conditions $x(0)\in[-25000,25000]$.
Closed-loop trajectories of a representative Lu subsystem under the synthesized controller are depicted in Figures~\ref{fig:lu}\protect\subref{fig:lu time} and \ref{fig:lu}\protect\subref{fig:lu 3D}.
	
	{\section{Conclusion}
	In this work, we offered an innovative compositional data-driven methodology for synthesizing controllers that guarantee GAS over large-scale interconnected networks with unknown mathematical models. By leveraging data from subsystem trajectories and ensuring input-to-state stability through ISS Lyapunov functions, our framework effectively designed local controllers for individual subsystems. The requirement of collecting only a single input-state trajectory for each subsystem up to a specified time horizon, coupled with utilizing small-gain compositional reasoning, enabled the construction of a CLF for the entire network, thus ensuring a GAS certificate. Importantly, our approach significantly reduced the computational complexity from a polynomial to a linear scale with respect to the number of subsystems. The efficacy and practicality of our method were validated through extensive case studies on various networks with diverse interconnection topologies. Developing a compositional data-driven technique for large-scale networks to encompass a broader class of nonlinear systems beyond polynomials is being investigated as future work.
	
	\bibliographystyle{ieeetran}
	\bibliography{biblio}
	
\end{document}